\documentclass[11pt,a4paper]{article}

\usepackage{hyperref}

\usepackage[fleqn]{amsmath}
\usepackage{amsfonts,amsthm,amssymb}

\usepackage{graphicx}

\usepackage{mathrsfs}
\usepackage{makeidx}
\usepackage{xypic}
\usepackage[nottoc]{tocbibind}
\usepackage{pstricks}
\usepackage{bbm}
\usepackage[arrow, matrix, curve]{xy}
\usepackage{epic}
\usepackage{eepic}
\usepackage{epsfig}

\numberwithin{equation}{section}

\usepackage{enumitem}

\xymatrixcolsep{0.5cm}
\xymatrixrowsep{0.5cm}

\newtheorem{theorem}{Theorem}[section]
\newtheorem{proposition}{Proposition}[section]
\newtheorem{lemma}{Lemma}[section]

\newtheorem{corollary}{Corollary}[section]

{\theoremstyle{remark}
\newtheorem{rem}{Remark}[section]}
\newcommand{\Id}{\mathrm{Id}}

\newcommand{\tr}{\operatorname{tr}}

\newcommand{\End}{\operatorname{End}}

\newcommand{\bra}[1]{\langle\,#1\,|}
\newcommand{\ket}[1]{|\,#1\,\rangle}

\newcommand{\moy}[1]{\langle\,#1\,\rangle}

\def\la{\lambda}

\newcommand{\mathsc}[1]{{\normalfont\textsc{#1}}}

\begin{document}

\begin{flushright}
LPENSL-TH-08/15
\end{flushright}

\bigskip \vspace{24pt}

\begin{center}
\begin{Large}
\textbf{The 8-vertex model with quasi-periodic boundary conditions}
\end{Large}

\vspace{45pt}

\begin{large}
{\bf G. Niccoli}\footnote{ENS Lyon; CNRS; Laboratoire de Physique, UMR 5672, Lyon France; giuliano.niccoli@ens-lyon.fr}
{\bf and V.~Terras}\footnote{Univ. Paris Sud; CNRS; LPTMS, UMR 8626, Orsay 91405 France; veronique.terras@lptms.u-psud.fr}
\end{large}

\vspace{45pt}

\today

\end{center}

\vspace{45pt}

\begin{abstract}
We study the inhomogeneous 8-vertex model (or equivalently the XYZ Heisenberg spin-1/2 chain) with all kinds of integrable quasi-periodic boundary conditions: periodic, $\sigma^x$-twisted, $\sigma^y$-twisted or $\sigma^z$-twisted.
We show that in all these cases but the periodic one with an even number of sites $\mathsf{N}$, the transfer matrix of the model is related, by the vertex-IRF transformation, to the transfer matrix of the dynamical $6$-vertex model with antiperiodic boundary conditions, which we have recently solved by means of Sklyanin's Separation of Variables (SOV) approach.
We show moreover that, in all the twisted cases, the vertex-IRF transformation is bijective.
This allows us to completely characterize, from our previous results on the antiperiodic dynamical 6-vertex model, the twisted 8-vertex transfer matrix spectrum (proving that it is simple) and eigenstates.
We also consider the periodic case for $\mathsf{N}$ odd. In this case we can define two independent vertex-IRF transformations, both not bijective, and by using them we show that the 8-vertex transfer matrix spectrum is doubly degenerate, and that it can, as well as  the corresponding eigenstates, also be completely characterized in terms of the spectrum and eigenstates of the dynamical 6-vertex antiperiodic transfer matrix.
In all these cases we can adapt to the 8-vertex case the reformulations of the dynamical 6-vertex transfer matrix spectrum and eigenstates that had been obtained by $T$-$Q$ functional equations, where the $Q$-functions are elliptic polynomials with twist-dependent quasi-periods.
Such reformulations enables one to characterize the 8-vertex  transfer matrix spectrum by the solutions of some Bethe-type equations, and to rewrite the corresponding eigenstates as the multiple action of some operators on a pseudo-vacuum state, in a similar way as in the algebraic Bethe ansatz framework.

\end{abstract}

\newpage


\section{Introduction}

The 8-vertex model is a two-dimensional model of statistical physics which generalizes the exactly solvable 6-vertex model, with two additional allowed configurations around a vertex \cite{Sut70,FanW70,Bax71a,Bax82L}. 
It is related to the completely anisotropic XYZ Heisenberg spin chain \cite{Bax71b},  a natural generalization of the XXZ spin chain. 
It is still exactly solvable in the sense that its associated $R$-matrix (i.e. the matrix of its local Boltzmann weights) satisfies the Yang-Baxter equation, and that the transfer matrices of the model form a one-parameter commuting family of operators. However, the charge through a vertex being not conserved, standard techniques such as Bethe ansatz cannot directly be applied.
In a nowadays famous paper \cite{Bax72}, Baxter managed to characterize the eigenvalues of the (periodic) transfer matrix by introducing a new object, the so-called $Q$-operator, whose eigenvalues satisfy, together with  the transfer matrix eigenvalues, a functional relation (the functional $T$-$Q$ relation). This new method nevertheless did not enabled him to obtain the eigenvectors. The latter  were constructed one year later \cite{Bax73a} by explicitly mapping the 8-vertex model onto a model of IRF (interaction-round-faces) type, the so-called ``8-vertex'' solid-on-solid (SOS) model (sometimes also called dynamical 6-vertex model or ABF model), which is solvable by Bethe ansatz. The relation between these two models is provided by the {\em vertex-IRF transformation}, which corresponds to a generalized gauge transformation between the $R$-matrices, and which can be extended to  the corresponding monodromy matrices and hence to the (periodic) transfer matrices and their eigenvectors.

Baxter's famous solution to the 8-vertex model has been at the origin of a long series of works \cite{Bax72a,JohKM73,Bax76,Bax77,FadT79,KluZ88,KluZ89,LasP98,Las02,Bax02,FabM03,Bax04,BooJMST05,FabM05,Fab07,FabM07,BazM07}. Let us mention for instance the reformulation of this solution in the framework of the Quantum Inverse Scattering Method (QISM) and of algebraic Bethe ansatz (ABA) \cite{FadT79,FelV96b}, some further study on the properties of the $Q$-operator or of the functional $T$-$Q$ equation \cite{Fab07,FabM07,BazM07,FabM09} and of the Bethe roots \cite{KluZ88,KluZ89,Bax02,BazM07}, the investigation of the algebraic structures (quantum groups) underlying the integrability of the model and its relation to the SOS model \cite{Skl82a,FodIJKMY94,Fro97,JimKOS99,BufRT12}, and some attempts to study the correlation functions \cite{JohKM73,JimMN93,LasP98,Las02,Shi04,BooJMST05}. 

It is however important to underline that Baxter's solution of the (finite-size) 8-vertex/XYZ model \cite{Bax73a} only applies to the diagonalization of transfer matrices constructed on a lattice with an {\em even} number of sites $\mathsf{N}$. This strong restriction comes from the fact that the vertex-IRF transformation has been used to relate the 8-vertex model with {\em periodic} boundary conditions to the SOS model with the {\em same type} of  (i.e. periodic) boundary conditions. It actually happens that the space of states of the exactly solvable SOS model with periodic boundary conditions is zero-dimensional  when $\mathsf{N}$ is odd, at least for general values of the crossing parameter of the model (in the root of unity case a few eigenstates can be constructed in that way, but they are nevertheless not sufficient to account for the whole 8-vertex space of states). Therefore, in spite of more than 40 years of works on the subject, a complete description of the 8-vertex/XYZ transfer matrix spectrum and eigenstates for $\mathsf{N}$ odd is still missing\footnote{Some analytic study of some solutions to the functional $T$-$Q$ equation in the case of $\mathsf{N}$ odd could nevertheless be performed at special values of the crossing parameter, see for instance \cite{FabM05,BazM05,BazM10}.}. 

The present paper aims at filling this gap. As first shown in \cite{Nic13a}, it happens that the transfer matrix of the periodic 8-vertex/XYZ model with $\mathsf{N}$ odd can be related, by means of a vertex-IRF transformation, to the transfer matrix of the SOS model with {\em antiperiodic} boundary conditions. A remarkable feature of this SOS antiperiodic model is that the dimension of the space of states is exactly  $2^\mathsf{N}$, i.e. is equal to the dimension of the space of states of the 8-vertex model, which is far from being the case when periodic boundary conditions are applied (even when $\mathsf{N}$ is even). This antiperiodic SOS model can be solved by means of Sklyanin's quantum version of the Separation of Variables (SOV) approach \cite{Skl85,Skl90,Skl92}, as recently shown in \cite{FelS99,Nic13a,LevNT15}. In fact, such a relation between the 8-vertex/XYZ transfer matrix and some appropriate version of the antiperiodic SOS transfer matrix also holds when the XYZ spin chain is subjected, for even or odd lattice size $\mathsf{N}$,  to an integrable boundary twist by some Pauli matrix $\sigma^\alpha$ ($\alpha=x,y$ or $z$). We hence also consider here these twisted cases. By studying the properties of the vertex-IRF transformation when acting on the corresponding space of states, we show that, in the periodic case for $\mathsf{N}$ odd, as well as in the twisted cases with $\mathsf{N}$ even or odd,  we are able to completely determine the 8-vertex transfer matrix spectrum and eigenstates by using the SOV characterization obtained in the antiperiodic SOS case \cite{Nic13a,LevNT15}. We hence obtain, as usual in the SOV framework, a complete description of the spectrum and eigenstates in terms of solutions of a set of discrete versions of Baxter's $T$-$Q$ equations at the inhomogeneity parameters of the model.
The results of \cite{LevNT15} enable us moreover to rewrite this characterization in terms of particular classes of solutions of some functional versions of Baxter's  $T$-$Q$ equation in which the $Q$-functions are elliptic polynomials with quasi-periods depending on the  boundary twist, i.e. in terms  of  the solutions of Bethe-type equations for the roots of these $Q$-functions.
In this framework the eigenstates can be obtained, quite similarly as in ABA, from the multiple action of some operator on a given pseudo-vacuum state, a picture that is more convenient than the initial SOV one for the consideration of the homogeneous and thermodynamic limits of the model.

The article is organized as follows.
In Section~\ref{sec-models} we recall the definition of the 8-vertex model in the QISM framework and its relation to the SOS (or dynamical 6-vertex) model by means of the vertex-IRF transformation.
In Section~\ref{sec-V-IRF} we show that the action on the space of states of the twisted 8-vertex transfer matrix is related to the action of the antiperiodic SOS transfer matrix, and that the vertex-IRF transformation which provides this relation is invertible in the case of a non-trivial twist.
This enables us in Section~\ref{sec-diag-8V} to completely determine the spectrum and eigenstates of the transfer matrix of the inhomogeneous 8-vertex/XYZ model in all the twisted cases, by using the SOV diagonalization of the antiperiodic SOS transfer matrix of \cite{LevNT15}. We also discuss in this section the rewriting of this characterization in terms of the solutions of Baxter's (homogeneous) functional $T$-$Q$ equation using the ansatz proposed in \cite{LevNT15}: this ansatz, which is proven to be complete at least in the case of $\mathsf{N}$ even, enables us to obtain the transfer matrix eigenvalues in terms of the solutions of some Bethe-type equations, and the eigenstates as multiple action of some operator on a particular pseudo-vacuum state.
Then, in Section~\ref{sec-per}, we consider the periodic case with $\mathsf{N}$ odd. In that case, the vertex-IRF transformation relating the action on the space of states of the periodic 8-vertex transfer matrix to the corresponding action of the antiperiodic SOS transfer matrix is no longer bijective, and therefore we use two versions of this vertex-IRF transformation to completely determine the periodic 8-vertex transfer matrix spectrum and eigenstates from the antiperiodic SOS ones.
In that case an ansatz for the solutions of Baxter's (homogeneous) $T$-$Q$ functional equation can still be proposed \cite{LevNT15} so as to rewrite the discrete SOV characterization of the eigenvalues and eigenstates in terms of solutions of Bethe-type equations, but  the completeness of this ansatz remains to be proven.
Two appendices complete this paper.
In Appendix A, we briefly recall the results of \cite{LevNT15} concerning the SOV solution of the antiperiodic dynamical 6-vertex model.
In Appendix B we present an alternative reformulation of the 8-vertex transfer matrix spectrum and eigenstates in terms of the elliptic polynomial solutions of some inhomogeneous version of Baxter's $T$-$Q$ equation, which also enables us to characterize the eigenstates in terms of solutions of some Bethe-type equations (although with some inhomogeneous term), and to obtain the eigenstates in some ABA-type form.  On the basis of the results of \cite{LevNT15}, such a  reformulation is complete for all the 8-vertex transfer matrices studied in this paper.

\section{The 8-vertex model and its relation to the dynamical 6-vertex model}
\label{sec-models}

In this section, we recall the definition of the 8-vertex model in the framework of the quantum inverse scattering method, as well as its relation to the dynamical 6-vertex model by means of Baxter's vertex-IRF transformation \cite{Bax73a,FelV96b}.

\subsection{The 8-vertex model in the QISM framework}

In the QISM framework, the 8-vertex model is associated to an elliptic $R$-matrix $R^{\mathsf{(8V)}}(\lambda )\in\End(\mathbb{C}^2\otimes \mathbb{C}^2)$ of the form
\begin{equation}\label{R-8V}
R^{\mathsf{(8V)}}(\lambda )
=
\begin{pmatrix}
  \mathsf{a}(\lambda) & 0 & 0 & \mathsf{d}(\lambda ) \\ 
0 & \mathsf{b}(\lambda ) & \mathsf{c}(\lambda ) & 0 \\ 
0 & \mathsf{c}(\lambda ) & \mathsf{b}(\lambda ) & 0 \\ 
\mathsf{d}(\lambda ) & 0 & 0 & \mathsf{a}(\lambda )
\end{pmatrix},
\end{equation}
where $\mathsf{a}(\lambda)$, $\mathsf{b}(\lambda)$, $\mathsf{c}(\lambda)$, $\mathsf{d}(\lambda)$, which parametrize the local Boltzmann weights of the model, are the following functions of the spectral parameter $\lambda$:
\begin{xalignat}{2}
 &\mathsf{a}(\lambda )
   =\frac{2\theta_4(\eta |2\omega )\, \theta_1(\lambda+\eta |2\omega )\,\theta_4(\lambda |2\omega )}
             {\theta_2(0|\omega )\, \theta_4(0|2\omega )},
 &\mathsf{b}(\lambda )
   =\frac{2\theta_4(\eta |2\omega)\, \theta_1(\lambda |2\omega )\, \theta_4(\lambda +\eta |2\omega )}
             {\theta_2(0|\omega )\, \theta_4(0|2\omega )}, \hspace{-2mm}
            \label{a-b} \\
 &\mathsf{c}(\lambda ) 
   =\frac{2\theta _1(\eta |2\omega )\, \theta_4(\lambda|2\omega )\, \theta_4(\lambda +\eta |2\omega )}
             {\theta_2(0|\omega )\, \theta_4(0|2\omega )},
  &\mathsf{d}(\lambda )
    =\frac{2\theta_1(\eta |2\omega)\,\theta_1(\lambda +\eta |2\omega )\, \theta _1(\lambda |2\omega )}
              {\theta_2(0|\omega )\, \theta _{4}(0|2\omega )}.\hspace{-2mm}
              \label{c-d}
\end{xalignat}
Here $\eta\in\mathbb{C}$ is a generic complex parameter which corresponds to the crossing parameter of the model.
The functions $\theta_j(\lambda|k\omega)$, $j=1,2,3,4$, $k=1,2$, denote the usual theta functions \cite{GraR07L} with quasi-periods $\pi$ and $k\pi\omega$ ($\Im\omega >0$).
In the following, we may simplify the notations for the theta functions with imaginary quasi-period $\pi\omega$ and write $\theta_j(\lambda)\equiv\theta_j(\lambda|\omega)$.

The $R$-matrix \eqref{R-8V} satisfies the Yang-Baxter equation on $\mathbb{C}^2\otimes \mathbb{C}^2\otimes\mathbb{C}^2$,
\begin{equation}\label{YB}
R_{12}^{\mathsf{(8V)}}(\lambda _{12})\, R_{13}^{\mathsf{(8V)}}(\lambda_{13})\, R_{23}^{\mathsf{(8V)}}(\lambda _{23})
=R_{23}^{\mathsf{(8V)}}(\lambda_{23})\, R_{13}^{\mathsf{(8V)}}(\lambda_{13})\, R_{12}^{\mathsf{(8V)}}(\lambda_{12}),
\end{equation}
as well as the following unitary and crossing symmetry relations on $\mathbb{C}^2\otimes \mathbb{C}^2$:
\begin{align}
   &R_{21}^{\mathsf{(8V)}}(-\lambda)\, R_{12}^{\mathsf{(8V)}}(\lambda) = \theta_1(-\lambda+\eta|\omega)\, \theta_1(\lambda+\eta|\omega)\, \Id_{12},\\
   &R_{12}^{\mathsf{(8V)}}(\lambda)\, \sigma_1^y\, \big[ R_{12}^{\mathsf{(8V)}}(\lambda-\eta)\big]^{t_1} \sigma_1^y =\theta_1(\lambda+\eta|\omega)\,\theta_1(\lambda-\eta|\omega)\, \Id_{12}.
\end{align}
Moreover, by using the quasi-periodicity properties of the theta functions, it is simple to show that the 8-vertex $R$-matrix \eqref{R-8V} also satisfies the identities
\begin{align}
   &R_{12}^{\mathsf{(8V)}}(\lambda +\pi )=-\sigma_{1}^{z}\, R_{12}^{\mathsf{(8V)}}(\lambda )\,\sigma _{1}^z,
   \label{R-period1}\\
   &R_{12}^{\mathsf{(8V)}}(\lambda +\pi \omega )
   =-e^{-i( 2\lambda +\pi\omega+\eta) }\, \sigma _{1}^{x}\, R_{12}^{\mathsf{(8V)}}(\lambda)\, \sigma _{1}^{x}.
   \label{R-period2}
\end{align}
As usual, the indices label the spaces of the tensor product on which the corresponding operator acts. $\sigma^\alpha$, $\alpha\in\{x,y,z\}$, stand for the Pauli matrices. We have also used the shorthand notation $\lambda_{ij}\equiv \lambda_i-\lambda_j$, 

The monodromy matrix $\mathsf{M}_{0}^{\mathsf{(8V)}}(\lambda )\equiv \mathsf{M}_{0,1\ldots\mathsf{N}}^{\mathsf{(8V)}}(\lambda;\xi_1,\ldots,\xi_{\mathsf{N}} )\in\End(V_0\otimes V_1\otimes \ldots\otimes V_{\mathsf{N}})$ of an inhomogenous model of size $\mathsf{N}$ with space of states $\mathbb{V}_\mathsf{N}=V_1\otimes V_2\ldots\otimes V_\mathsf{N}\simeq (\mathbb{C}^2)^{\otimes\mathsf{N}}$ is defined as the following ordered product of $R$-matrices: 
\begin{equation}\label{mon-1}
  \mathsf{M}_{0}^{\mathsf{(8V)}}(\lambda )
  = R_{0\mathsf{N}}^{\mathsf{(8V)}}(\lambda -\xi_{\mathsf{N}})\cdots 
              R_{02}^{\mathsf{(8V)}}(\lambda -\xi_2)\, R_{01}^{\mathsf{(8V)}}(\lambda -\xi_{1}).
\end{equation}
Here $V_0\simeq\mathbb{C}^2$ denotes the auxiliary space, $V_n\simeq\mathbb{C}^2$ the local space at site $n$, and $\xi_n\in\mathbb{C}$ is the $n^\mathrm{th}$ inhomogeneity parameter.
In this paper we shall suppose that the inhomogeneity parameters are generic, or at least that they satisfy the following conditions
\begin{align}
   &\xi_j+\frac{\eta}{2}\notin \pi\mathbb{Z}+\pi\omega\mathbb{Z},\label{cond-inh0}\\
   &\forall\epsilon\in  \{-1,0,1\}, \quad \xi _{a}- \xi_{b}+\epsilon\eta \notin \pi\mathbb{Z}+\pi\omega\mathbb{Z}
 \quad \text{if}\  a\neq b.
    \label{cond-inh}
\end{align}
The monodromy matrix \eqref{mon-1} can be represented as a $2\times 2$ matrix on the auxiliary space $V_0$,
\begin{equation}
  \mathsf{M}_{0}^{\mathsf{(8V)}}(\lambda )= \begin{pmatrix}
                 \mathsf{A}^{\mathsf{(8V)}}(\lambda ) & \mathsf{B}^{\mathsf{(8V)}}(\lambda ) \\ 
                 \mathsf{C}^{\mathsf{(8V)}}(\lambda ) & \mathsf{D}^{\mathsf{(8V)}}(\lambda )
                 \end{pmatrix}_{\! [0]},
                 \label{mon-2}
\end{equation}
and its operator entries $\mathsf{A}^{\mathsf{(8V)}}(\lambda )$, $\mathsf{B}^{\mathsf{(8V)}}(\lambda )$, $\mathsf{C}^{\mathsf{(8V)}}(\lambda )$, $\mathsf{D}^{\mathsf{(8V)}}(\lambda )$ satisfy the Yang-Baxter commutation relations following from the quadratic relation on ${V}_0\otimes {V}_{0'}\otimes\mathbb{V}_\mathsf{N}$
\begin{equation} \label{RTT}
R_{0 0'}^{\mathsf{(8V)}}(\lambda _{12})\, \mathsf{M}_{0}^{\mathsf{(8V)}}(\lambda_{1})\,
\mathsf{M}_{0'}^{\mathsf{(8V)}}(\lambda _{2})
=\mathsf{M}_{0'}^{\mathsf{(8V)}}(\lambda _{2})\, \mathsf{M}_{0}^{\mathsf{(8V)}}(\lambda _{1})\,
R_{0 0'}^{\mathsf{(8V)}}(\lambda _{12}).
\end{equation}

The purpose of this paper is to explain how to characterize the spectrum and construct the complete set of eigenstates of the transfer matrix of the model with particular quasi-periodic boundary conditions. 
From \eqref{RTT} and from the fact that, for $\alpha\in\{z,x,y\}$, the $R$-matrix \eqref{R-8V} satisfies the discrete symmetry property $[R(\lambda), \sigma^\alpha\otimes\sigma^\alpha]=0$, it is easy to see that the transfer matrices
\begin{equation}\label{transfer}
   \mathsf{T}^{\mathsf{(8V)}}_\mathsf{(x,y)}(\lambda )
   =\tr\left[    \mathsf{M}_\mathsf{(x,y)}^\mathsf{(8V)}(\lambda)\right],
\end{equation}
where
\begin{equation}\label{twist-mon}
   \mathsf{M}_\mathsf{(x,y)}^\mathsf{(8V)}(\lambda)=\mathsf{K}^{\mathsf{(x,y)}}\, \mathsf{M}^{\mathsf{(8V)}}(\lambda ),   \qquad
   \text{with}\quad
   \mathsf{K^{(x,y)}}= (\sigma^{x})^\mathsf{y}\, (\sigma^{z})^\mathsf{x},
\end{equation}
define, for each choice of $\mathsf{(x,y)}\in\{0,1\}^2$, a one-parameter family of commuting operators:
\begin{equation}
   \big[ \mathsf{T}^{\mathsf{(8V)}}_\mathsf{(x,y)}(\lambda ), \mathsf{T}^{\mathsf{(8V)}}_\mathsf{(x,y)}(\mu ) \big]=0,
   \qquad \lambda,\mu\in\mathbb{C}.
\end{equation}
The different cases,
\begin{align}
  &\mathsf{T}^{\mathsf{(8V)}}_\mathsf{(0,0)}(\lambda )
  = \mathsf{A}^{\mathsf{(8V)}}(\lambda ) + \mathsf{D}^{\mathsf{(8V)}}(\lambda ) ,\\
  &\mathsf{T}^{\mathsf{(8V)}}_\mathsf{(1,0)}(\lambda )
  = \mathsf{A}^{\mathsf{(8V)}}(\lambda ) - \mathsf{D}^{\mathsf{(8V)}}(\lambda ) ,\\
  &\mathsf{T}^{\mathsf{(8V)}}_\mathsf{(0,1)}(\lambda )
  = \mathsf{B}^{\mathsf{(8V)}}(\lambda ) + \mathsf{C}^{\mathsf{(8V)}}(\lambda ) ,\\
  &\mathsf{T}^{\mathsf{(8V)}}_\mathsf{(1,1)}(\lambda )
  = \mathsf{B}^{\mathsf{(8V)}}(\lambda ) - \mathsf{C}^{\mathsf{(8V)}}(\lambda ) ,
\end{align}
%
correspond here to different types of quasi-periodic boundary conditions for the model, that we call respectively periodic, $\sigma^z$-twisted, antiperiodic (or $\sigma^x$-twisted), and twisted antiperiodic (proportional to $\sigma^y$-twisted) boundary conditions.
The logarithmic derivative of the transfer matrix gives, in the homogeneous limit, the Hamiltonian of the XYZ spin-1/2 Heisenberg chain,
\begin{equation}
   \frac{\partial \log \mathsf{T}^{\mathsf{(8V)}}_\mathsf{(x,y)}(\lambda )}{\partial\lambda}\bigg|_{\substack{\lambda=0\\ \xi_n=0}}
   = H_{\mathrm{XYZ}}=\frac{1}{2}\sum_{n=1}^\mathsf{N}\Big\{ J_x \sigma_n^x\sigma_{n+1}^x+J_y \sigma_n^y\sigma_{n+1}^y +J_z \sigma_n^z\sigma_{n+1}^z\Big\} +\frac{1}{2} J_0,
\end{equation}
with the corresponding quasi-periodic boundary conditions:
\begin{equation}
  \sigma_{\mathsf{N}+1}^\alpha=\mathsf{K}_1^{\mathsf{(x,y)}}\, \sigma_1^\alpha\, \big(\mathsf{K}_1^{\mathsf{(x,y)}} \big)^{-1}
  \qquad \text{for any}\quad \alpha=x,y,z.
\end{equation}
Here
\begin{alignat}{2}
  & J_x= \frac{\theta_1'(0|2\omega)}{\theta_1(0|2\omega)} \left[ \frac{\theta_4(\eta|2\omega)}{\theta_1(\eta|2\omega)}+\frac{\theta_1(\eta|2\omega)}{\theta_4(\eta|2\omega)}\right],
  \qquad
  &J_z=\frac{\theta_1'(\eta|2\omega)}{\theta_1(\eta|2\omega)}-\frac{\theta_4'(\eta|2\omega)}{\theta_4(\eta|2\omega)},
  \\
  &J_y= \frac{\theta_1'(0|2\omega)}{\theta_1(0|2\omega)} \left[ \frac{\theta_4(\eta|2\omega)}{\theta_1(\eta|2\omega)}-\frac{\theta_1(\eta|2\omega)}{\theta_4(\eta|2\omega)}\right],
  \qquad
  &J_0=\frac{\theta_1'(\eta|2\omega)}{\theta_1(\eta|2\omega)}+\frac{\theta_4'(\eta|2\omega)}{\theta_4(\eta|2\omega)}.
\end{alignat}

Before closing this subsection, we finally recall two important properties issued from the study of the 8-vertex Yang-Baxter algebra: 
\begin{itemize}
\item the inversion relation for the monodromy matrix \eqref{mon-1}:
\begin{equation}  \label{Inv-8v-M}
   \mathsf{M}_{0}^{\mathsf{(8V)}}(\lambda )\cdot 
   \sigma_{0}^{y}\,\big[ \mathsf{M}_{0}^{\mathsf{(8V)}}(\lambda -\eta )\big]^{t_{0}}\sigma _{0}^{y}
   =\mathrm{det}_q \mathsf{M}^{\mathsf{(8V)}}(\lambda ),
\end{equation}
where $\mathrm{det}_q \mathsf{M}^{\mathsf{(8V)}}(\lambda )$ is the so-called quantum determinant, which is a central element of the 8-vertex Yang-Baxter algebra:
\begin{align}
\mathrm{det}_{q}\mathsf{M}^{\mathsf{(8V)}}(\lambda )
 &=\text{\textsc{a}}(\lambda )\text{\textsc{d}}(\lambda -\eta ) \nonumber\\
& =  \mathsf{A}^{\mathsf{(8V)}}(\lambda )\, \mathsf{D}^{\mathsf{(8V)}}(\lambda -\eta )
                 -\mathsf{B}^{\mathsf{(8V)}}(\lambda )\,\mathsf{C}^{\mathsf{(8V)}}(\lambda -\eta )
\nonumber\\
& = \mathsf{D}^{\mathsf{(8V)}}(\lambda )\,\mathsf{A}^{\mathsf{(8V)}}(\lambda -\eta )
      -\mathsf{C}^{\mathsf{(8V)}}(\lambda )\, \mathsf{B}^{\mathsf{(8V)}}(\lambda -\eta ) ,
      \label{det-q}
\end{align}
with
\begin{equation}\label{a-d}
\text{\textsc{a}}(\lambda )\equiv \prod_{n=1}^{\mathsf{N}}\theta (\lambda-\xi _{n}+\eta|\omega ),
\quad 
\text{\textsc{d}}(\lambda )\equiv \text{\textsc{a}}(\lambda -\eta );
\end{equation}
\item the solution of the quantum inverse problem \cite{KitMT99,MaiT00,KitKMNST07} which enables one to express any elementary local operator  $X_n\in\End(V_n)$ at site $n$ in terms of the entries of the monodromy matrix \eqref{mon-1} 
or of its inverse as
\begin{align}
     X_n
     &=\prod_{k=1}^{n-1}\mathsf{T}_\mathsf{(0,0)}^\mathsf{(8V)}(\xi_k)\cdot
     \tr_0\big[ \mathsf{M}_0^\mathsf{(8V)}(\xi_n)\, X_0\big] \cdot
     \prod_{k=1}^{n} \big[ \mathsf{T}_\mathsf{(0,0)}^\mathsf{(8V)}(\xi_k) \big]^{-1},
     \label{inv-pb1}\\
     &=\prod_{k=1}^{n}\mathsf{T}_\mathsf{(0,0)}^\mathsf{(8V)}(\xi_k)\cdot
     \frac{\tr_0 \big[ \sigma_{0}^{y}\, \mathsf{M}_{0}^{\mathsf{(8V)}}(\xi_n -\eta )^{t_{0}}\,\sigma _{0}^{y}\, X_0\big]}{\det_q\mathsf{M}^\mathsf{(8V)}(\xi_n)}\cdot
     \prod_{k=1}^{n-1} \big[ \mathsf{T}_\mathsf{(0,0)}^\mathsf{(8V)}(\xi_k) \big]^{-1}.
     \label{inv-pb2}
\end{align}
\end{itemize}

\subsection{Elementary properties of the 8-vertex transfer matrices}

The QISM algebraic framework presented in the last subsection enables one to obtain without further study some properties of the quasi-periodic transfer matrices \eqref{transfer} that we gather in the next proposition.

\begin{proposition}
The quasi-periodic 8-vertex transfer matrices \eqref{transfer} satisfy the quasi-periodicity properties
\begin{align}
   & \mathsf{T}^{\mathsf{(8V)}}_\mathsf{(x,y)}(\lambda+\pi )=(-1)^{\mathsf{N}+\mathsf{y}}\, \mathsf{T}^{\mathsf{(8V)}}_\mathsf{(x,y)}(\lambda ),
   \label{per-transfer1}\\
   &\mathsf{T}^{\mathsf{(8V)}}_\mathsf{(x,y)}(\lambda+\pi\omega )
   =(-e^{-2i\lambda-i\pi\omega})^\mathsf{N}\, e^{2i [\sum_{k=1}^\mathsf{N}\xi_k-\frac{\mathsf{N}}{2}\eta+\mathsf{x}\frac{\pi}{2} ]}\, \mathsf{T}^{\mathsf{(8V)}}_\mathsf{(x,y)}(\lambda ).
   \label{per-transfer2}
\end{align}
They moreover satisfy the following identities when evaluated at the inhomogeneity parameters:
\begin{align}
   &\mathsf{T}^{\mathsf{(8V)}}_\mathsf{(x,y)}(\xi_n)\, \mathsf{T}^{\mathsf{(8V)}}_\mathsf{(x,y)}(\xi_n-\eta )
   =(-1)^\mathsf{x+y}\, \mathrm{det}_q \mathsf{M}^{\mathsf{(8V)}}(\xi_n),
   \qquad
   n\in\{1,\ldots,\mathsf{N}\},\label{q-det-ID}
   \displaybreak[0]\\
  &\prod_{n=1}^\mathsf{N}\mathsf{T}^{\mathsf{(8V)}}_\mathsf{(x,y)}(\xi_n)=\prod_{n=1}^\mathsf{N}\mathsc{a}(\xi_n)\, \prod_{n=1}^\mathsf{N}\mathsf{K}^\mathsf{(x,y)}_n.\label{global-ID}
\end{align}
\end{proposition}

\begin{proof}
The quasi-periodicity properties \eqref{R-period1}-\eqref{R-period2} of the  8-vertex $R$-matrix lead to the following identities for the $\mathsf{(x,y)}$-twisted monodromy matrix \eqref{twist-mon}:
\begin{align}
  &\mathsf{M}_\mathsf{(x,y)}^\mathsf{(8V)}(\lambda +\pi )
  =(-1)^\mathsf{N+y}
  \sigma _{0}^{z}\, \mathsf{M}_\mathsf{(x,y)}^\mathsf{(8V)}(\lambda)\,\sigma _{0}^{z},
 \\
  &\mathsf{M}_\mathsf{(x,y)}^\mathsf{(8V)}(\lambda +\pi \omega )
  =(-1)^\mathsf{N+x}\, e^{-i( 2\lambda +\pi \omega+\eta) \mathsf{N}}\, e^{2i \sum_{k=1}^{\mathsf{N}}\xi_k }\,
  \sigma_0^{x}\,\mathsf{M}_\mathsf{(x,y)}^\mathsf{(8V)}(\lambda)\,\sigma_0^{x}. 
\end{align}
The quasi-periodicity properties \eqref{per-transfer1} and \eqref{per-transfer2} of the $\mathsf{(x,y)}$-twisted 8-vertex transfer matrix then follow by means of the cyclicity of the trace.

From the reconstruction formulae of local operators \eqref{inv-pb1}-\eqref{inv-pb2}, it is easy to prove \cite{KitKMNST07,Nic13a} the
annihilation identities
\begin{align}
&\mathsf{A}^{\mathsf{(8V)}}(\xi _{n})\, \mathsf{A}^{\mathsf{(8V)}}(\xi_{n}-\eta)
 =\mathsf{D}^{\mathsf{(8V)}}(\xi _{n})\, \mathsf{D}^{\mathsf{(8V)}}(\xi _{n}-\eta)=0,  \label{annih-1}
  \\
&\mathsf{B}^{\mathsf{(8V)}}(\xi _{n})\, \mathsf{B}^{\mathsf{(8V)}}(\xi_{n}-\eta)
 =\mathsf{C}^{\mathsf{(8V)}}(\xi _{n})\, \mathsf{C}^{\mathsf{(8V)}}(\xi _{n}-\eta)=0,  \label{annih-2}
\end{align}
as well as the exchange identities
\begin{align}
&\mathsf{A}^{\mathsf{(8V)}}(\xi _{n})\, \mathsf{D}^{\mathsf{(8V)}}(\xi_{n}-\eta) 
=-\mathsf{C}^{\mathsf{(8V)}}(\xi _{n})\, \mathsf{B}^{\mathsf{(8V)}}(\xi _{n}-\eta),  \label{excha-1} \\
&\mathsf{D}^{\mathsf{(8V)}}(\xi _{n})\, \mathsf{A}^{\mathsf{(8V)}}(\xi_{n}-\eta) 
=-\mathsf{B}^{\mathsf{(8V)}}(\xi _{n})\, \mathsf{C}^{\mathsf{(8V)}}(\xi _{n}-\eta).  \label{excha-2}
\end{align}
%
By using the annihilation identities \eqref{annih-1}-\eqref{annih-2} we get
\begin{equation}
    \mathsf{T}_{(\mathsf{x,y})}^{\mathsf{(8V)}}(\xi _{n})\,\mathsf{T}_{(\mathsf{x,y})}^{\mathsf{(8V)}}(\xi _n-\eta)
    =(-1)^{\mathsf{y}}\Big[
     \mathsf{A}^{\mathsf{(8V)}}(\xi _{n})\,\mathsf{D}^{\mathsf{(8V)}}(\xi_{n}-\eta)
    +\mathsf{D}^{\mathsf{(8V)}}(\xi _{n})\, \mathsf{A}^{\mathsf{(8V)}}(\xi _{n}-\eta)\Big],
\end{equation}
for $(\mathsf{x,y})=\,(0,0)\,,(1,0)$, and 
\begin{equation}
 \mathsf{T}_{(\mathsf{x,y})}^{\mathsf{(8V)}}(\xi _{n})\,\mathsf{T}_{(\mathsf{x,y})}^{\mathsf{(8V)}}(\xi _{n}-\eta)
 =(-1)^{\mathsf{x}}\Big[
  \mathsf{B}^{\mathsf{(8V)}}(\xi _{n})\,\mathsf{C}^{\mathsf{(8V)}}(\xi_{n}-\eta)
  +\mathsf{C}^{\mathsf{(8V)}}(\xi _{n})\,\mathsf{B}^{\mathsf{(8V)}}(\xi _{n}-\eta)\Big],
\end{equation}
for $(\mathsf{x,y})=\,(0,1)\,,(1,1)$. These relations lead to the so-called inversion formula \eqref{q-det-ID} for the transfer matrix once we use the exchange relations \eqref{excha-1}-\eqref{excha-2}.

Finally, the identities \eqref{global-ID} are also a trivial consequence of the reconstruction formula for the local operators $\mathsf{K}^\mathsf{(x,y)}_n$ once we recall that the product of periodic transfer matrices evaluated at the inhomogeneity parameters along the chain is just the product of all $\mathsc{a}(\xi_n)$.
\end{proof}

As noticed in  \cite{Nic13a} for the periodic case,  one can use the above properties to get a preliminary description of the transfer matrix spectrum. Indeed, the relations \eqref{per-transfer1}-\eqref{per-transfer2} imply that all eigenvalues $\mathsf{t}(\lambda)$ of $\mathsf{T}^{\mathsf{(8V)}}_\mathsf{(x,y)}(\lambda )$ should satisfy the quasi-periodicity properties
\begin{align}
 & \mathsf{t}(\lambda+\pi)=(-1)^{\mathsf{N}+\mathsf{y}}\, {\mathsf{t}}(\lambda),
 \label{periodt-1}\\
 &\mathsf{t}(\lambda+\pi\omega)
 = (-e^{-2i\lambda-i\pi\omega})^\mathsf{N}\, e^{2i [\sum_{k=1}^\mathsf{N}\xi_k-\frac{\mathsf{N}}{2}\eta+\mathsf{x}\frac{\pi}{2}]} \, {\mathsf{t}}(\lambda),
 \label{periodt-2}
\end{align}
as well as the following identities when evaluated at the inhomogeneity parameters:
\begin{align}
   &\mathsf{t}(\xi_n)\, \mathsf{t}(\xi_n-\eta)
   =( -1) ^{\mathsf{x}+\mathsf{y}}\,\mathsc{a}(\xi _{n})\,\mathsc{d}(\xi _{n}-\eta),
\qquad
\forall n\in \{1,\ldots,\mathsf{N}\},
\label{eq-quadr-8V}\displaybreak[0]\\
  &\prod_{n=1}^\mathsf{N} \mathsf{t}(\xi_n)=(\pm 1)^\mathsf{x+y+xy}\prod_{n=1}^\mathsf{N}\mathsc{a}(\xi_n).
  \label{eq-prodt}
\end{align}
Let us mention here that, after the paper \cite{Nic13a}, an ansatz based on a systematic use of this kind of identities was proposed to describe the spectrum of several integrable model, by some modified Baxter's type functional equation (see \cite{CaoCYSW14} for the periodic 8-vertex model).
However, it is important to stress that the above properties are only {\em necessary} conditions that have to be satisfied by the transfer matrix eigenvalues, but that, without further information about the nature of the spectrum, they are {\em a priori} not sufficient to ensure its complete characterization: we do not know at this stage whether all the solutions to \eqref{periodt-1}-\eqref{eq-prodt}  are indeed transfer matrix eigenvalues.

Hence, to proceed further, we shall use Baxter's  vertex-IRF transformation \cite{Bax73a} to relate the 8-vertex model to the SOS (or dynamical 6-vertex) one. In this context we shall be able to explicitly construct, in the twisted case or in the periodic case with $\mathsf{N}$ odd, the transfer matrix eigenstates by SOV. This essential step will notably allow us to select the true eigenvalues among the solutions to \eqref{periodt-1}-\eqref{eq-prodt} as those for which we can construct nonzero eigenstates.

\subsection{Vertex-IRF transformation from 8-vertex to dynamical 6-vertex model}

Although the form of the 8-vertex $R$-matrix \eqref{R-8V} does not {\em a priori} allow for the direct resolution of the model by Bethe ansatz, Baxter nevertheless managed to construct the eigenstates of the periodic 8-vertex transfer matrix with an even number of sites by relating them to the eigenstates of the transfer matrix of another model of statistical physics of solid-on-solid (SOS) type \cite{Bax73a}. This SOS model, also sometimes called ABF \cite{AndBF84} or dynamical 6-vertex model, describes interactions of a height variable around the faces of a two-dimensional square lattice. It is directly solvable by Bethe ansatz and its variants has been widely studied (see for instance \cite{DatJMO86,KunY88,PeaS88,PeaS89,DatJKM90,FelV96b,Ros09,PakRS08,LevT13a,LevT13b,LevT14a}). The relation between the two models is provided by a generalized gauge transformation connecting their respective $R$-matrices, called {\em vertex-IRF} transformation, and which can be written in the following form:
\begin{equation}\label{vertex-IRF}
R_{12}^{\mathsf{(8V)}}(\lambda _{12})\, S_1(\lambda_1|t )\, S_2(\lambda_2| t +\eta \sigma_1^{z})
=S_2(\lambda_2| t )\, S_1(\lambda_1| t +\eta \sigma _2^z)\, R_{12}(\lambda _{12}|t).
\end{equation}
Here $t$ is an additional parameter called dynamical parameter. In the language of the SOS model, it corresponds to the value of the height variable at a particular site of the model. The $R$-matrix $R(\lambda|t)$ of the SOS (or dynamical 6-vertex) model depends on both the spectral and dynamical parameters and satisfies a modified version of the Yang-Baxter relation called dynamical Yang-Baxter equation \cite{GerN84,Fel95}:
\begin{equation}\label{DYBE}
  R_{12}(\lambda_{12}|t+\eta\sigma_3^z)\, R_{13}(\lambda_{13}|t)\,
  R_{23}(\lambda_{23}|t+\eta\sigma_1^z)\\
  =R_{23}(\lambda_{23}|t)\, R_{13}(\lambda_{13}|t+\eta\sigma_2^z)\,
  R_{12}(\lambda_{12}|t).
\end{equation}

In this paper we consider the following solution of \eqref{DYBE}:
\begin{equation}\label{R-6VD}
  R(\lambda|t)=
  \begin{pmatrix}
  a(\lambda) & 0 & 0 & 0 \\
  0 & e^{i\mathsf{y}\eta}\, b(\lambda|t) & e^{i\mathsf{y}\lambda}\,c(\lambda|t) & 0 \\
  0 & e^{-i\mathsf{y}\lambda}\,c(\lambda|-t) & e^{-i\mathsf{y}\eta}\, b(\lambda|-t) & 0\\
  0 & 0 & 0 & a(\lambda)
  \end{pmatrix} ,
\end{equation}
with $\mathsf{y}\in\{0,1\}$ and
\begin{equation}\label{def-abc}
   a(\lambda)=\theta(\lambda+\eta), \quad 
   b(\lambda|t)=\frac{\theta(\lambda)\,\theta(t+\eta)}{\theta(t)}, \quad
   c(\lambda|t)=\frac{\theta(\eta)\,\theta(t+\lambda)}{\theta(t)}.
\end{equation}
The dynamical gauge transformation \eqref{vertex-IRF} between \eqref{R-8V} and \eqref{R-6VD} is provided by the following $2\times 2$ numerical matrix,
\begin{equation}\label{mat-S}
S(\lambda | t )= e^{i\mathsf{y}\frac{t}{2}}\,
\begin{pmatrix}
e^{-i\mathsf{y}\frac{\lambda}{2}}\,\theta_{2}(-\lambda + t |2\omega) & 
e^{i\mathsf{y}\frac{\lambda}{2}}\,\theta_{2}(\lambda + t |2\omega) \\ 
e^{-i\mathsf{y}\frac{\lambda}{2}}\,\theta_{3}(-\lambda + t |2\omega) & 
e^{i\mathsf{y}\frac{\lambda}{2}}\,\theta_{3}(\lambda + t |2\omega)
\end{pmatrix},
\end{equation}
which also depends on both the spectral and dynamical parameters.

The vertex-IRF transformation \eqref{vertex-IRF} can be extended to a relation between the corresponding monodromy matrices:
\begin{equation}\label{v-IRF-mon}
   \mathsf{M}_{0}^{\mathsf{(8V)}}(\lambda )\, S_{0}(\lambda | t )\, S_{q}(t+\eta \sigma_{0}^{z})
   =S_{q}(t)\, S_{0}(\lambda | t +\eta \mathsf{S})\, \mathsf{M}_{0}(\lambda |t ),
\end{equation}
where $\mathsf{S}$ is the total $z$-component of the spin,
\begin{equation}\label{total-spin}
   \mathsf{S} = \sum_{j=1}^{\mathsf{N}}\sigma_j^{z} ,
\end{equation}
and where $S_q(t)\equiv S_q(\xi_1,\ldots,\xi_\mathsf{N}| t)$, with $q\equiv 1\,2\ldots \mathsf{N}$, is defined as
\begin{equation}  \label{S_q}
   S_{q}( t ) = S_{1}(\xi _{1}| t )\, S_2(\xi_2| t+\eta \sigma_1^z) \ldots S_\mathsf{N}(\xi _{\mathsf{N}}| t +\eta \sum_{a=1}^{\mathsf{N}-1}\sigma_{a}^{z}).
\end{equation}
In \eqref{v-IRF-mon}, $\mathsf{M}_{0}(\lambda |t )$ denotes the (periodic) dynamical 6-vertex monodromy matrix, defined on $V_0\otimes\mathbb{V}_\mathsf{N}$ as
\begin{align}
\mathsf{M}_{0}(\lambda |t )
&\equiv 
R_{0 \mathsf{N}}(\lambda -\xi _{\mathsf{N}}|t +\eta \sum_{a=1}^{\mathsf{N}-1}\sigma _{a}^{z})\cdots R_{0 2}(\lambda-\xi_2|t+\eta\sigma_1^z)\, R_{0 1}(\lambda -\xi _{1}|t)\nonumber\\
&\equiv 
\begin{pmatrix}
\mathsf{A}(\lambda |t ) & \mathsf{B}(\lambda |t ) \\ 
\mathsf{C}(\lambda |t ) & \mathsf{D}(\lambda |t )
\end{pmatrix}_{\! [0]}
. \label{mon-6VD}
\end{align}
The latter satisfies, together with \eqref{R-6VD}, a dynamical quadratic relation of the form
\begin{equation}\label{RTT-dyn}
R_{00'}(\lambda _{0 0'}|t+\eta \mathsf{S})\,
\mathsf{M}_{0}(\lambda _{0}|t )\,
\mathsf{M}_{0'}(\lambda _{0'}|t +\eta \sigma _{0}^{z})
=\mathsf{M}_{0'}(\lambda _{0'}| t )\,
\mathsf{M}_{0}(\lambda _{0}|t+\eta \sigma _{0'}^{z})\,
R_{00'}(\lambda _{0 0'}|t ).
\end{equation}

The algebraic Bethe ansatz for the dynamical 6-vertex model with periodic boundary conditions has been formulated in \cite{FelV96b} from the representation theory of the elliptic quantum group studied in \cite{FelV96a}. The study of the antiperiodic model in the framework of the quantum separation of variables approach has been performed in \cite{FelS99,Nic13a,LevNT15}.
In particular, in \cite{LevNT15}, slightly different variants of the antiperiodic model have been considered, depending on different global shifts of the dynamical parameter $t$.
We shall use here the results of \cite{LevNT15} to explicitly construct, by means of the vertex-IRF transformation, the eigenvectors and eigenvalues of the 8-vertex transfer matrices \eqref{transfer} in the quasi-periodic cases with $\mathsf{(x,y)}\not=(0,0)$ and $\mathsf{N}$ even or odd, and in the periodic case with 
 $\mathsf{(x,y)}=(0,0)$ and $\mathsf{N}$ odd.


\section{Vertex-IRF transformation and quasi-periodic transfer matrices in left and right representation spaces}
\label{sec-V-IRF}

The search for a convenient gauge transformation of the 8-vertex monodromy matrix simplifying the analysis of the 8-vertex transfer matrix spectrum has
naturally led to the introduction of the dynamical parameter $t$.
The space of states of the gauge transformed model (the dynamical 6-vertex model) hence corresponds to a representation space of the dynamical Yang-Baxter algebra. In this section, we describe this dynamical-spin representation space and show that, in certain conditions, the vertex-IRF transformation defines an isomorphism between a particular subspace of this dynamical-spin space and the $2^\mathsf{N}$-dimensional pure spin space of states of the 8-vertex model.
This enables us to completely characterize the action on this space of states of the quasi-periodic 8-vertex transfer matrices \eqref{transfer} for $\mathsf{(x,y)}\not=(0,0)$ in terms of the action of 
the transfer matrix of the dynamical 6-vertex model with antiperiodic boundary conditions.

\subsection{Dynamical-spin and pure spin representation spaces}

As described in \cite{Nic13a,LevNT15}, it is convenient, so as to simplify the commutation relations issued from \eqref{RTT-dyn}, to extend our spin operator algebra by introducing some dynamical operators $\tau$ and $\mathsf{T}_{\tau }^\pm$ which commute with the spin operators and which satisfy the commutation relations
\begin{equation}
\mathsf{T}_{\tau }^\pm \tau =(\tau \pm \eta )\mathsf{T}_{\tau}^\pm,
\label{Dyn-op-comm}
\end{equation}
and to define a new monodromy matrix incorporating these dynamical operators as
\begin{equation}\label{mon-op}
\mathcal{M}_{0}(\lambda )
\equiv \mathsf{M}_{0}(\lambda |\tau )\, \mathsf{T}_{\tau }^{\sigma_{0}^{z}}
\equiv \begin{pmatrix}
\mathcal{A}(\lambda  ) & \mathcal{B}(\lambda  ) \\ 
\mathcal{C}(\lambda  ) & \mathcal{D}(\lambda  )
\end{pmatrix}_{\! [0]}.
\end{equation}
The advantage of this formulation is that the operator entries of \eqref{mon-op} satisfy simpler commutation relations than the operators entries of \eqref{mon-6VD}, given by the quadratic relation
\begin{equation}\label{RTT-op}
R_{00'}(\lambda _{00'}|\tau +\eta \mathsf{S})\,
\mathcal{M}_{0}(\lambda _{0} )\,\mathcal{M}_{0'}(\lambda _{0'} )
 =\mathcal{M}_{0'}(\lambda _{0'} )\,\mathcal{M}_{0}(\lambda _{0} )\,R_{00'}(\lambda _{00'}|\tau ).
\end{equation}
It also satisfies the following inversion formula:
\begin{equation}\label{inv-mon}
  \mathcal{M}_0(\lambda)\cdot \sigma_0^y\,\mathcal{M}_0(\lambda-\eta)^{t_0}\,\sigma_0^y
  =e^{-i\mathsf{y}\eta\mathsf{S}}\frac{\theta(\tau)}{\theta(\tau+\eta\mathsf{S})}\,\mathrm{det}_q M(\lambda),
\end{equation}
in terms of the quantum determinant $\mathrm{det}_q M(\lambda)=\mathsc{a}(\lambda)\,\mathsc{d}(\lambda-\eta)$.

The operator entries of \eqref{mon-op} are then though of as acting on some dynamical-spin space $\mathbb{D}_{\mathsf{(6VD)},\mathsf{N}}\equiv \mathbb{V}_\mathsf{N}\otimes\mathbb{D}$, where $\mathbb{D}$ is an infinite-dimensional representation space of the dynamical operator algebra \eqref{Dyn-op-comm} with left (covectors) and right (vectors) $\tau $-eigenbasis respectively defined as
\begin{equation}
\langle t(a)|\equiv \langle t(0)|\mathsf{T}_{\tau }^{-a},\qquad
|t(a)\rangle \equiv \mathsf{T}_{\tau }^{a}|t(0)\rangle ,\qquad
\forall a\in \mathbb{Z},  \label{t-dyn-sp}
\end{equation}
such that
\begin{equation} \label{Dynamical-spectrum-1}
\langle t(a)|\tau =t(a)\langle t(a)|,\quad\
\tau |t(a)\rangle =t(a)|t(a)\rangle ,
 \qquad 
t(a) \equiv -\eta a+t_0,\quad
\forall a\in \mathbb{Z},
\end{equation}
with the normalization $\langle t(a)|t(b)\rangle=\delta _{a,b},\ \forall a,b\in \mathbb{Z}$.
In this paper, as in \cite{LevNT15}, we fix the value of the global shift $t_0$ to be, in terms of $\mathsf{N}$, $\mathsf{y}$, and of some additional parameter $\mathsf{x}\in\{0,1\}$,
\begin{equation}\label{Dynamical-spectrum-2}
 t_0=  -\frac{\eta }{2}\mathsf{N}  +\mathsf{x}\frac{\pi }{2}+\mathsf{y}\frac{\pi}{2}\omega,
\end{equation}
with the condition $(\mathsf{x},\mathsf{y})\not=(0,0)$ if $\mathsf{N}$ is even, and we denote the corresponding left and right representation spaces\footnote{As in \cite{LevNT15}, we may denote by a subscript $\mathcal{R}$ the representation spaces for the dynamical and spin operators (i.e. for instance $\mathbb{D}_{\mathsf{(6VD)},\mathsf{N}}^\mathcal{R}\equiv \mathbb{D}_{\mathsf{(6VD)},\mathsf{N}}$, $\mathbb{V}_\mathsf{N}^\mathcal{R}\equiv \mathbb{V}_\mathsf{N}, \ldots$), and by a subscript $\mathcal{L}$ their restricted dual spaces that we shall call left representation spaces.} as $\mathbb{D}^{\mathcal{L}/\mathcal{R}}\equiv \mathbb{D}_{(\mathsf{x},\mathsf{y}),\mathsf{N}}^{\mathcal{L}/\mathcal{R}}$.
A $\sigma_n^z$-eigenbasis in the local spin space $V_n^{\mathcal{L}/\mathcal{R}}$, $n\in \{1,\ldots,\mathsf{N}\}$, is given by the states $\langle n,h_{n}|$ (resp. $ |n,h_{n}\rangle$), $h_n\in\{0,1\}$, such that
\begin{equation}
\langle n,h_{n}|\, \sigma _{n}^{z}=(1-2h_{n})\, \langle n,h_{n}|,
\qquad\text{resp.}\quad
\sigma _{n}^{z}\, |n,h_{n}\rangle =(1-2h_{n})\, |n,h_{n}\rangle ,
\end{equation}
with $\langle n,h_{n}|n,h_{n}^{\prime }\rangle =\delta _{h_{n},h_{n}^{\prime}}$.
Hence a natural basis of $\mathbb{D}_{\mathsf{(6VD)},\mathsf{N}}^{\mathcal{L}}$ (respectively of $\mathbb{D}_{\mathsf{(6VD)},\mathsf{N}}^{\mathcal{R}}$) is provided by the vectors
\begin{equation}\label{dyn-spin-basis}
(\otimes _{n=1}^{\mathsf{N}}\langle n,h_{n}|)\otimes \langle t(a)|,
\qquad
\text{resp.}\quad
(\otimes _{n=1}^{\mathsf{N}}|n,h_{n}\rangle )\otimes |t(a)\rangle ,
\end{equation}
obtained by tensoring  common eigenstates of the commuting operators $\tau $ and $\sigma _{n}^{z}$, $1\le n\le \mathsf{N}$, with the following scalar product:
\begin{equation}
\big(\otimes _{n=1}^{\mathsf{N}}|n,h_{n}\rangle \otimes |t(a)\rangle ,
\otimes_{n=1}^{\mathsf{N}}|n,h_{n}^{\prime }\rangle \otimes |t(a^{\prime })\rangle\big)
=\delta _{a,a^{\prime }}\prod_{n=1}^{\mathsf{N}}\delta_{h_{n},h_{n}^{\prime }}.
\end{equation}

Let us define the operator
\begin{equation}
\mathsf{S}_{\tau }\equiv \eta \mathsf{S}+2\tau \in \End(\mathbb{D}_{\mathsf{(6VD)},\mathsf{N}}).
\end{equation}
For each $r\in\mathbb{Z}$, we denote with $\mathbb{\bar{D}}_{\mathsf{(6VD)},\mathsf{N}}^{(r,\mathcal{L}/\mathcal{R}) }$ the $2^{\mathsf{N}}$-dimensional left and right linear eigenspaces of $\mathsf{S}_{\tau }$ corresponding to the eigenvalue $2r\eta +\mathsf{x}\pi+\mathsf{y}\pi\omega $, which are respectively  generated by the vectors
\begin{alignat}{2}
  &\big( \otimes _{n=1}^{\mathsf{N}}\langle n,h_{n}|\big) \otimes \langle t_{r,\mathbf{h}}|,
  \qquad 
  &&\mathbf{h}\equiv (h_1,\ldots, h_\mathsf{N})\in \{0,1\}^\mathsf{N},
   \label{DyS-basis-L}\\
  &\big( \otimes _{n=1}^{\mathsf{N}} | n,h_{n}\rangle\big) \otimes | t_{r,\mathbf{h}}\rangle,
  \qquad 
  &&\mathbf{h}\equiv (h_1,\ldots, h_\mathsf{N})\in \{0,1\}^\mathsf{N}, \label{DyS-basis-R}
\end{alignat}
where
\begin{equation}\label{t_h}
t_{r,\mathbf{h}} = -\frac{\eta }{2}\mathsf{s}_{\mathbf{h}}+\mathsf{x}\frac{\pi }{2}+\mathsf{y}\frac{\pi}{2}\omega+r\eta
=t_0+\sum_{k=1}^\mathsf{N} h_k +r\eta,
\qquad\text{with}\quad \mathsf{s}_{\mathbf{h}}= \sum_{k=1}^\mathsf{N}(1-2h_{k}).
\end{equation}
We recall that
\begin{proposition}
 For each $r\in\mathbb{Z}$, the $2^\mathsf{N}$-dimensional vector space $\mathbb{\bar{D}}_{\mathsf{(6VD)},\mathsf{N}}^{(r,\mathcal{L}/\mathcal{R}) }$ is invariant under the action of the operators $\mathsf{A}(\lambda|\tau)$, $\mathsf{D}(\lambda|\tau)$, $\mathcal{B}(\lambda)$, $\mathcal{C(\lambda)}$.
\end{proposition}
In particular, $\mathbb{\bar{D}}_{\mathsf{(6VD)},\mathsf{N}}^{(r=0) }$ corresponds to the physical space of states of the SOS model with antiperiodic boundary conditions, as studied in \cite{Nic13a,LevNT15}.
We recall here that the definition of this model depends on the values of the two parameters $\mathsf{x},\mathsf{y}\in\{0,1\}$ appearing in the $R$-matrix \eqref{R-6VD} and in the global shift \eqref{Dynamical-spectrum-2} of the dynamical parameter, hence we may sometimes call it the $\mathsf{(x,y)}$-dynamical 6-vertex model.

We also define the following homomorphism $\mathbf{P}^{\mathcal{L/R}}$ from the representation space $\vspace{-1.5mm}\mathbb{D}_{\mathsf{(6VD)},\mathsf{N}}^{\mathcal{L/R}}$ of the dynamical Yang-Baxter algebra to the pure spin space of states $\mathbb{V}_\mathsf{N}^{\mathcal{L/R}}$ of the XYZ model by its action on the basis vectors \eqref{dyn-spin-basis}:
\begin{align}
  &\mathbf{P}^{\mathcal{R}}:
  \big( \otimes _{n=1}^{\mathsf{N}} | n,h_{n}\rangle\big) \otimes | t(a)\rangle
  \mapsto \big( \otimes _{n=1}^{\mathsf{N}} | n,h_{n}\rangle\big) ,
  \label{def-PR}\\
  &\mathbf{P}^{\mathcal{L}}:
  \big( \otimes _{n=1}^{\mathsf{N}} \langle n,h_{n} |\big) \otimes \langle t(a) |
  \mapsto \big( \otimes _{n=1}^{\mathsf{N}} \langle n,h_{n} | \big).
  \label{def-PL}
\end{align}
By definition, we have that, for each $\bra{\mathbf{v}}\in\mathbb{D}_{\mathsf{(6VD)},\mathsf{N}}^\mathcal{L}$, respectively $\ket{\mathbf{v}}\in\mathbb{D}_{\mathsf{(6VD)},\mathsf{N}}^\mathcal{R}$, 
\begin{equation}\label{shift-P}
  \mathbf{P}^{\mathcal{L}} \big( \bra{\mathbf{v}} \mathsf{T}_\tau^\pm \big)
  =  \mathbf{P}^{\mathcal{L}} \big( \bra{\mathbf{v}}  \big),
  \quad \text{resp.}\quad
  \mathbf{P}^{\mathcal{R}} \big(  \mathsf{T}_\tau^\pm \ket{\mathbf{v}}\big)
  =  \mathbf{P}^{\mathcal{R}} \big( \ket{\mathbf{v}}  \big).
\end{equation}
Moreover, for each $r\in\mathbb{Z}$, the restriction $\mathbf{P}^{(r,\mathcal{L/R})}$ of  $\mathbf{P}^\mathcal{L/R}$ to the subspace  $\mathbb{\bar{D}}_{\mathsf{(6VD)},\mathsf{N}}^{(r,\mathcal{L}/\mathcal{R}) }$ of $\mathbb{D}_{\mathsf{(6VD)},\mathsf{N}}^{\mathcal{L/R}}$ defines an isomorphism from $\mathbb{\bar{D}}_{\mathsf{(6VD)},\mathsf{N}}^{(r,\mathcal{L}/\mathcal{R}) }$ to $\mathbb{V}_\mathsf{N}^{\mathcal{L/R}}$, and its action on any vector  $\bra{\mathbf{v}}\in\mathbb{\bar{D}}_{\mathsf{(6VD)},\mathsf{N}}^{(r,\mathcal{L})}$, respectively $\ket{\mathbf{v}}\in\mathbb{\bar{D}}_{\mathsf{(6VD)},\mathsf{N}}^{(r,\mathcal{R})}$, is given by
\begin{equation}\label{act-Pr}
 \mathbf{P}^{(r,\mathcal{L})} (\bra{\mathbf{v} })
  = \bra{\mathbf{v} }\left( \sum_{a=0}^\mathsf{N} \ket{ t_{r,\mathbf{0}} +\eta a} \right),
  \ \ \text{resp.}\ \
  \mathbf{P}^{(r,\mathcal{R})} (\ket{\mathbf{v} })
  = \left( \sum_{a=0}^\mathsf{N} \langle\, t_{r,\mathbf{0}} +\eta a \,|\right) \ket{\mathbf{v} }.
\end{equation}
In particular, the dynamical-spin space of states $\mathbb{\bar{D}}_{\mathsf{(6VD)},\mathsf{N}}^{(0,\mathcal{L/R})}$ of the dynamical 6-vertex model with antiperiodic boundary conditions is isomorphic, by means of the above mapping,  to the pure spin quantum space of states $\mathbb{V}_\mathsf{N}^{\mathcal{L/R}}$ of the XYZ model. This is a clear advantage with respect to the study of the periodic XYZ model by means of its relation to the periodic dynamical 6-vertex model, since the latter has a space of states which has not the same dimension as the space of states of the former.

\begin{rem}
In the following, we shall simply, in accordance with the definition \eqref{def-PR}-\eqref{def-PL}, use the following notation: for any  $\ket{\mathbf{v}}\in\mathbb{D}_{\mathsf{(6VD)},\mathsf{N}}^{(r,\mathcal{R})}$, respectively $\bra{\mathbf{v}}\in\mathbb{D}_{\mathsf{(6VD)},\mathsf{N}}^{(r,\mathcal{L})}$,
\begin{equation}
   \mathbf{P}^{(r,\mathcal{R})}\big(\ket{\mathbf{v}}\big)
   =\mathbf{P}^{( r)}\,\ket{\mathbf{v}},
   \qquad  \text{resp.}\quad
   \mathbf{P}^{(r,\mathcal{L})}\big(\bra{\mathbf{v}}\big)=\bra{\mathbf{v}}\,\big[\mathbf{P}^{( r )}\big]^{-1}.
\end{equation}
\end{rem}

\subsection{The vertex-IRF transformation as an isomorphism of vector spaces}

We shall now prove that, for $\mathsf{(x,y)}\not=(0,0)$, the vertex-IRF transformation \eqref{S_q} that relates the monodromy matrices of the 8-vertex and  $\mathsf{(x,y)}$-dynamical 6-vertex models is bijective.
More precisely,

\begin{proposition}
\label{prop-V-IRF-iso}
Let $r\in\mathbb{Z}$.
On any vector of $\ket{\mathbf{v}}\in\mathbb{\bar{D}}_{\mathsf{(6VD)},\mathsf{N}}^{(r,\mathcal{R})}$, 
one has
\begin{equation}
\mathbf{P}^{\mathcal{R}}\big(S_{q}(\tau )\ket{\mathbf{v}}\big)
=\mathbf{S}^{(r )}\,\mathbf{P}^{(r )}\, \ket{\mathbf{v}}, 
\label{P-S}
\end{equation}
where the action of the operator $\mathbf{S}^{(r )}\in\End(\mathbb{V}_\mathsf{N})$ is defined on the local spin basis vectors of $\mathbb{V}_\mathsf{N}$ as
\begin{align}
   &\mathbf{S}^{(r )}\Big(\underset{n=1}{\overset{\mathsf{N}}{\otimes}}\!\ket{n,h_{n}}\Big)
      = 
      S_q(t_{r,\mathbf{h}})
      \ \underset{n=1}{\overset{\mathsf{N}}{\otimes}}\!\ket{n,h_{n}} , \label{S-spin-R}
\end{align}
where $S_q(t_{r,\mathbf{h}})$ stands for the operator \eqref{S_q} evaluated at the value \eqref{t_h}.

When $(\mathsf{x},\mathsf{y})\not=(0,0)$, $\mathbf{S}^{(r )}$ is an automorphism of $\mathbb{V}_\mathsf{N}$.
\end{proposition}

\begin{rem}
Note that, even if not explicitly underlined, this operator $\mathbf{S}^{(r )}$ depends on the values of $\mathsf{x}$ and $\mathsf{y}$ through the definition \eqref{mat-S} of $S(\lambda|t)$ and  the value \eqref{t_h} of $t_{r,\mathbf{h}}$.
\end{rem}

\begin{proof}
It is enough to consider the action on the generic elements \eqref{DyS-basis-R} of the dynamical-spin basis of $\mathbb{\bar{D}}_{\mathsf{(6VD)},\mathsf{N}}^{(r,\mathcal{R})}$ to prove the
formula \eqref{P-S} and the characterization \eqref{S-spin-R}.

Let us now define, for $j\in\{1,\ldots,N\}$, the operators $\mathbf{S}^{(r,j)}\in \End(\mathbb{V}_{\mathsf{N}})$ by their action on the generic basis elements $\otimes_{n=1}^{\mathsf{N}}|n,h_{n}\rangle $:
\begin{multline}\label{Sj}
  \mathbf{S}^{(r,j)}\, \Big(\underset{n=1}{\overset{\mathsf{N}}{\otimes}}\!\ket{n,h_{n}}\Big)
  \equiv
  S_{j}\Big(\xi_{j}\Big| t_{r,\mathbf{h}}+\eta \sum_{n=1}^{j-1}\sigma_{n}^{z}\Big)\,
  S_{j+1}\Big(\xi _{j+1} \Big| t_{r,\mathbf{h}}+\eta \sum_{n=1}^{j}\sigma_{n}^{z}\Big)\ldots
  \\
  \ldots
  S_{\mathsf{N}}\Big(\xi _{\mathsf{N}} \Big| t_{r,\mathbf{h}}+\eta\sum_{n=1}^{\mathsf{N}-1}\sigma _{n}^{z}\Big)\,
  \Big(\underset{n=1}{\overset{\mathsf{N}}{\otimes}}\!\ket{n,h_{n}}\Big) ,
\end{multline}
and $\mathbf{S}^{(r,\mathsf{N}+1)}\equiv\mathrm{Id}$.
By definition, we have that $\mathbf{S}^{(r)} = \mathbf{S}^{(r,1)}$, and we want to show by induction on $j$ that, for $\mathsf{(x,y)}\not=0$ and for all $j\in\{1,\ldots,\mathsf{N}+1\}$, $\mathbf{S}^{(r,j)}$ is an isomorphism.

It is clearly the case for $\mathbf{S}^{(r,\mathsf{N}+1)}=\mathrm{Id}$.
Let us therefore assume that, for some $j\in\{1,\ldots,\mathsf{N}\}$, $\mathbf{S}^{(r,j+1)}$ is an isomorphism.  Note that, by definition, $\mathbf{S}^{(r,j+1)}$ acts trivially on  $V_{1}\otimes\ldots\otimes V_{j}$, so that we can write
\begin{equation}
    \mathbf{S}^{(r,j+1)}\, \Big(\underset{n=1}{\overset{\mathsf{N}}{\otimes}}\!\ket{n,h_{n}}\Big)
    = \Big(\underset{n=1}{\overset{j}{\otimes}}\!\ket{n,h_{n}}\Big)\otimes
       \ket{[h_{j+1},\ldots,h_\mathsf{N}]_{\mathbf{S}^{(r,j+1)}} },
\end{equation}
where, by hypothesis, the vectors $\ket{[h_{j+1},\ldots,h_\mathsf{N}]_{\mathbf{S}^{(r,j+1)}} }$ form, for $(h_{j+1},\ldots,h_\mathsf{N})\in\{0,1\}^{\mathsf{N}-j}$, a basis of $ V_{j+1}\otimes\ldots\otimes V_\mathsf{N}$ (and $\bra{[h_{j+1},\ldots,h_\mathsf{N}]_{\mathbf{S}^{(r,j+1)}} }$ will denote the elements of its dual basis).
Then the action of $\mathbf{S}^{(r,j)}$ on the local spin basis vectors of $\mathbb{V}_\mathsf{N}$ is given as
\begin{align*}
    &\mathbf{S}^{(r,j)}\,
    \Big(\underset{n=1}{\overset{\mathsf{N}}{\otimes}}\!\ket{n,h_{n}}\Big)\\
    &\quad=\Big(\underset{n=1}{\,\overset{j-1}{\otimes}}\ket{n,h_{n}}\Big)
    \otimes \bigg[
    S_{j}\Big(\xi _{j} \Big| t_{r,\mathbf{h}}+\eta \sum_{n=1}^{j-1}(1-2h_n)\Big) \ket{j,h_{j}} \bigg]
     \otimes
       \ket{[h_{j+1},\ldots,h_\mathsf{N}]_{\mathbf{S}^{(r,j+1)}} }
       \\
    &\quad=\Big(\underset{n=1}{\,\overset{j-1}{\otimes}}\ket{n,h_{n}}\Big)\\
    &\quad
    \otimes 
       \begin{pmatrix}
e^{i\frac{\mathsf{y}}{2}(\hat{t}_{r,\mathbf{h}}^{\, (j)}-\frac{\eta}{2}-\xi_j)}\,
\theta_{2}( \hat{t}_{r,\mathbf{h}}^{\, (j)}-\frac{\eta}{2}-\xi_j |2\omega) 
& 
e^{i\frac{\mathsf{y}}{2}(\hat{t}_{r,\mathbf{h}}^{\, (j)}+\frac{\eta}{2}+\xi_j)}\,
\theta_{2}( \hat{t}_{r,\mathbf{h}}^{\, (j)}+\frac{\eta}{2}+\xi_j |2\omega) 
\\ 
e^{i\frac{\mathsf{y}}{2}(\hat{t}_{r,\mathbf{h}}^{\, (j)}-\frac{\eta}{2}-\xi_j)}\,
\theta_{3}(\hat{t}_{r,\mathbf{h}}^{\, (j)} -\frac{\eta}{2}-\xi_j |2\omega) 
& 
e^{i\frac{\mathsf{y}}{2}(\hat{t}_{r,\mathbf{h}}^{\, (j)}+\frac{\eta}{2}+\xi_j)}\,
\theta_{3}( \hat{t}_{r,\mathbf{h}}^{\, (j)}+\frac{\eta}{2}+\xi_j |2\omega)
\end{pmatrix}_{\!\! [j]}
\ket{j,h_j}
\\
 &\hspace{10cm}\otimes
       \ket{[h_{j+1},\ldots,h_\mathsf{N}]_{\mathbf{S}^{(r,j+1)}} },
\end{align*}
where we have defined
\begin{equation}\label{thatj}
   \hat{t}_{r,\mathbf{h}}^{\, (j)}
   \equiv \frac{\eta }{2}\sum_{k=1}^{j-1}(1-2h_{k})+\frac{\eta }{2}\sum_{k=j+1}^{\mathsf{N}}(2h_{k}-1)
   +\mathsf{x}\frac{\pi }{2}+\mathsf{y}\frac{\pi}{2}\omega+r\eta.
\end{equation}
Note that, for generic $\eta$ (i.e. incommensurable to $\pi$ and $\pi\omega$) and $(\mathsf{x},\mathsf{y})\not=(0,0)$, $\hat{t}_{r,\mathbf{h}}^{\, (j)}\notin \pi\mathbb{Z}+\pi\omega\mathbb{Z}$.
Since
\begin{multline}\label{invert-a-cond}
\det     \begin{pmatrix}
e^{i\frac{\mathsf{y}}{2}(\hat{t}_{r,\mathbf{h}}^{\, (j)}-\frac{\eta}{2}-\xi_j)}\,
\theta_{2}( \hat{t}_{r,\mathbf{h}}^{\, (j)}-\frac{\eta}{2}-\xi_j |2\omega) 
& 
e^{i\frac{\mathsf{y}}{2}(\hat{t}_{r,\mathbf{h}}^{\, (j)}+\frac{\eta}{2}+\xi_j)}\,
\theta_{2}( \hat{t}_{r,\mathbf{h}}^{\, (j)}+\frac{\eta}{2}+\xi_j |2\omega) 
\\ 
e^{i\frac{\mathsf{y}}{2}(\hat{t}_{r,\mathbf{h}}^{\, (j)}-\frac{\eta}{2}-\xi_j)}\,
\theta_{3}(\hat{t}_{r,\mathbf{h}}^{\, (j)} -\frac{\eta}{2}-\xi_j |2\omega) 
& 
e^{i\frac{\mathsf{y}}{2}(\hat{t}_{r,\mathbf{h}}^{\, (j)}+\frac{\eta}{2}+\xi_j)}\,
\theta_{3}( \hat{t}_{r,\mathbf{h}}^{\, (j)}+\frac{\eta}{2}+\xi_j |2\omega)
\end{pmatrix}\\
=e^{i\mathsf{y} \hat{t}_{r,\mathbf{h}}^{\, (j)}}\,
\theta (\xi_{j}+\eta /2)\, \theta (\hat{t}_{r,\mathbf{h}}^{\, (j)})
\neq 0,
\end{multline}
for generic $\eta$ and $(\mathsf{x},\mathsf{y})\not=(0,0)$, then $\mathbf{S}^{(r,j)}$ is also an isomorphism.
Indeed, let us consider the equation, for some vector $\ket{\mathbf{v} }\in\mathbb{V}_\mathsf{N}$,
\begin{equation}
   \mathbf{S}^{(r,j)}\, \ket{\mathbf{v} } =0,
   \qquad\text{with}\quad 
   \ket{\mathbf{v} }=\sum_{\mathbf{h}\in\{0,1\}^{\mathsf{N}}} c_{\mathbf{h} }\, \Big(\underset{n=1}{\overset{\mathsf{N}}{\otimes}}\!\ket{n,h_{n}}\Big).
\end{equation}
It reduces to an equation in the $j$-th local quantum space when we consider the
following matrix elements, for any $(k_1,\ldots,k_{j-1},k_{j+1},\ldots k_\mathsf{N})\in\{0,1\}^{\mathsf{N}-1}$,
\begin{align*}
  0 &= \Bigg(\underset{n=1}{\overset{j-1}{\otimes}}\!\bra{n,k_{n}}\otimes  \bra{[k_{j+1},\ldots,k_\mathsf{N}]_{\mathbf{S}^{(r,j+1)}} } \Bigg)\, \mathbf{S}^{(r,j)} \, \ket{\mathbf{v} }
  \\
        &=\sum_{h_{j}=0}^{1} c_{(k_1,\ldots,k_{j-1},h_j,k_{j+1},\ldots,k_{\mathsf{N}})}
         \\
        &
        \times\begin{pmatrix}
e^{i\frac{\mathsf{y}}{2}(\hat{t}_{r,\mathbf{k}}^{\, (j)}-\frac{\eta}{2}-\xi_j)}\,
\theta_{2}( \hat{t}_{r,\mathbf{k}}^{\, (j)}-\frac{\eta}{2}-\xi_j |2\omega) 
& 
e^{i\frac{\mathsf{y}}{2}(\hat{t}_{r,\mathbf{k}}^{\, (j)}+\frac{\eta}{2}+\xi_j)}\,
\theta_{2}( \hat{t}_{r,\mathbf{k}}^{\, (j)}+\frac{\eta}{2}+\xi_j |2\omega) 
\\ 
e^{i\frac{\mathsf{y}}{2}(\hat{t}_{r,\mathbf{k}}^{\, (j)}-\frac{\eta}{2}-\xi_j)}\,
\theta_{3}(\hat{t}_{r,\mathbf{k}}^{\, (j)} -\frac{\eta}{2}-\xi_j |2\omega) 
& 
e^{i\frac{\mathsf{y}}{2}(\hat{t}_{r,\mathbf{k}}^{\, (j)}+\frac{\eta}{2}+\xi_j)}\,
\theta_{3}( \hat{t}_{r,\mathbf{k}}^{\, (j)}+\frac{\eta}{2}+\xi_j |2\omega)
\end{pmatrix}_{\!\! [j]}
\ket{j,h_j},
\end{align*}
which, from the condition \eqref{invert-a-cond}, can be
satisfied if and only if 
\begin{equation}
c_{(k_1,\ldots,k_{j-1},0,k_{j+1},\ldots,k_{\mathsf{N}})}=c_{(k_1,\ldots,k_{j-1},1,k_{j+1},\ldots,k_{\mathsf{N}})}=0.
\end{equation}
Being  $(k_1,\ldots,k_{j-1},k_{j+1},\ldots k_\mathsf{N})\in\{0,1\}^{\mathsf{N}-1}$ completely arbitrary, the statement follows.
\end{proof}

\subsection{The relation between the quasi-periodic 8-vertex and the antiperiodic dynamical 6-vertex transfer matrices}

Using the vertex-IRF transformation \eqref{P-S}, it is possible to relate the $\mathsf{(x,y)}$-twisted 8-vertex transfer matrix \eqref{transfer} 
to the transfer matrix
\begin{equation}\label{6VD-antiT}
   \overline{\mathcal{T}}(\lambda)= \mathcal{B}(\lambda)+\mathcal{C}(\lambda)
\end{equation}
of the {\em antiperiodic} dynamical 6-vertex model defined on the dynamical-spin space of states $\mathbb{\bar{D}}_{\mathsf{(6VD)},\mathsf{N}}^{(0)}$ with the {\em same} choice
  of $\mathsf{(x,y)}$.
More precisely, we have the following result.

\begin{proposition}\label{th-act-T}
The $\mathsf{(x,y)}$-twisted 8-vertex transfer matrix \eqref{transfer}, combined with the vertex-IRF transformation, has the following action on the states $\ket{\mathbf{v}}$ of $\mathbb{\bar{D}}_{\mathsf{(6VD)},\mathsf{N}}^{(0,\mathcal{R})}$,
\begin{equation}
\mathsf{T}_{\mathsf{(x,y)}}^{\mathsf{(8V)}}(\lambda )\, S_{q}(\tau )\, \ket{\mathbf{v}}
  =(-1)^{\mathsf{x}}\, i^{\mathsf{xy}}\,
    \big[ S_{q}(\tau -\eta )\,\mathsf{C}(\lambda |\tau -\eta )
           +S_{q}(\tau +\eta )\,\mathsf{B}(\lambda |\tau+\eta )\big] \ket{\mathbf{v}},  
\label{right-act-T}
\end{equation}
which can be rewritten in terms of the action of the transfer matrix \eqref{6VD-antiT} of the antiperiodic $\mathsf{(x,y)}$-dynamical 6-vertex model as
\begin{equation}\label{right-act-T2}
   \mathsf{T}_{\mathsf{(x,y)}}^{\mathsf{(8V)}}(\lambda )\, \mathbf{S}^{(0)}\,\mathbf{P}^{(0)}\, \ket{\mathbf{v}}
  =(-1)^{\mathsf{x}}\, i^{\mathsf{xy}}\,  \mathbf{S}^{(0)}\,\mathbf{P}^{(0)}\,
  \overline{\mathcal{T}}(\lambda)\, \ket{\mathbf{v}}.
\end{equation}
\end{proposition}

Since the vertex-IRF transformation is an isomorphism of vector spaces  when $\mathsf{(x,y)}\not=(0,0)$, the relation \eqref{right-act-T2} can in that case be formulated directly at the operator level, as stated in the following corollary.

\begin{corollary}\label{cor-rel-T}
When $\mathsf{(x,y)}\not=(0,0)$, the $\mathsf{(x,y)}$-twisted 8-vertex transfer matrix \eqref{transfer} can be expressed in terms of the antiperiodic $\mathsf{(x,y)}$-dynamical 6-vertex transfer matrix \eqref{6VD-antiT} as
\begin{equation}\label{rel-T}
   \mathsf{T}_{\mathsf{(x,y)}}^{\mathsf{(8V)}}(\lambda )
   =(-1)^{\mathsf{x}}\, i^{\mathsf{xy}}\, \mathbf{S}^{(0)}\,\mathbf{P}^{(0)}\,
  \overline{\mathcal{T}}(\lambda)\, \big[ \mathbf{P}^{(0)}\big]^{-1}\,\big[ \mathbf{S}^{(0)}\big]^{-1}.
\end{equation}
\end{corollary}

It may be convenient, here and in the following, to define, through the isomorphism $\mathbf{P}^{(0)}$, some analogs of the dynamical Yang-Baxter operators which act on the pure spin space $\mathbb{V}_\mathsf{N}$, in particular
%
\begin{align}
   &\mathsf{B^{(6VD)}}(\lambda)=\mathbf{P}^{(0)}\, \mathcal{B}(\lambda)\, \big[ \mathbf{P}^{(0)}\big]^{-1},
   \qquad\quad
     \mathsf{C^{(6VD)}}(\lambda)=\mathbf{P}^{(0)}\, \mathcal{C}(\lambda)\, \big[ \mathbf{P}^{(0)}\big]^{-1},
     \label{def-BC}\\
   &\mathsf{\overline{T}^{(6VD)}}(\lambda)=\mathsf{B^{(6VD)}}(\lambda)+\mathsf{C^{(6VD)}}(\lambda).
   \label{dyn-transfer}
\end{align}
The relation \eqref{rel-T} for  $\mathsf{(x,y)}\not=(0,0)$ can therefore simply be rewritten as
\begin{equation}\label{rel-Tbis}
   \mathsf{T}_{\mathsf{(x,y)}}^{\mathsf{(8V)}}(\lambda )
   =(-1)^{\mathsf{x}}\, i^{\mathsf{xy}}\, \mathbf{S}^{(0)}\, \mathsf{\overline{T}^{(6VD)}}(\lambda)  \,\big[ \mathbf{S}^{(0)}\big]^{-1}.
\end{equation}

\begin{rem}
The action of the quasi-periodic 8-vertex transfer matrix to the left, i.e. the analog of \eqref{right-act-T}-\eqref{right-act-T2} on the states $\bra{\mathbf{v}}$ of $\mathbb{\bar{D}}_{\mathsf{(6VD)},\mathsf{N}}^{(0,\mathcal{L})}$, follows directly from \eqref{rel-T} in the case  $\mathsf{(x,y)}\not=(0,0)$.
\end{rem}

We now turn to the proof of Proposition~\ref{th-act-T}, for which we use the following Lemma:

\begin{lemma}
For any vector $\ket{\mathbf{v}}\in\mathbb{\bar{D}}_{\mathsf{(6VD)},\mathsf{N}}^{(0,\mathcal{R})}$, one has
\begin{align}
   &\mathbf{P}^\mathcal{R} \big( S_{q}(\tau -\eta )\,\mathsf{C}(\lambda |\tau -\eta )\,
   \ket{\mathbf{v}}\big)
   =  \mathbf{S}^{(0)}\, \mathbf{P}^{(0)}\, \mathcal{C}(\lambda)\,
   \ket{\mathbf{v}},
   \label{actC1}\\
   &\mathbf{P}^\mathcal{R}\big(S_{q}(\tau +\eta )\,\mathsf{B}(\lambda |\tau +\eta )\,
   \ket{\mathbf{v}} \big)
   =  \mathbf{S}^{(0)}\, \mathbf{P}^{(0)}\, \mathcal{B}(\lambda)\,
   \ket{\mathbf{v}}.
   \label{actB1}
\end{align}
\end{lemma}

\begin{proof}
We recall that, if $\ket{\mathbf{v}}\in\mathbb{\bar{D}}_{\mathsf{(6VD)},\mathsf{N}}^{(0,\mathcal{R})}$, then $\mathcal{C}(\lambda) \ket{\mathbf{v}}\in\mathbb{\bar{D}}_{\mathsf{(6VD)},\mathsf{N}}^{(0,\mathcal{R})}$, so that, from \eqref{P-S},
\begin{align}
   \mathbf{S}^{(0)}\, \mathbf{P}^{(0)}\, \mathcal{C}(\lambda)\,
   \ket{\mathbf{v}}
   &= \mathbf{P}^\mathcal{R}\big( S_q(\tau)\, \mathcal{C}(\lambda)\,
   \ket{\mathbf{v}} \big)
   \nonumber\\
   &= \mathbf{P}^\mathcal{R}\big( S_q(\tau)\, \mathsf{C}(\lambda|\tau)\, \mathsf{T}_\tau^+\,
   \ket{\mathbf{v}} \big)
   \nonumber\\
   &= \mathbf{P}^\mathcal{R}\big( S_q(\tau-\eta)\, \mathsf{C}(\lambda|\tau-\eta)\, 
   \ket{\mathbf{v}} \big),
\end{align}
where we have used successively \eqref{Dyn-op-comm} and \eqref{shift-P}.
\eqref{actB1} can be proven similarly.
\end{proof}

{\it Proof of Proposition~\ref{th-act-T}.}
Let us show \eqref{right-act-T}, which is a generalization to the $\mathsf{(x,y)}$-case of  Lemma 3.3 of \cite{Nic13a}.
To prove \eqref{right-act-T}, we rewrite, similarly as in \cite{Nic13a}, the gauge transformation \eqref{v-IRF-mon} under the following equivalent form:
\begin{multline}\label{v-IRF1}
  S_{0}(\lambda +\eta |\tau )\, S_{q}(\tau +\eta \sigma _{0}^{z})\,
  \begin{pmatrix}
     \mathsf{D}(\lambda |\tau +\eta ) & -\mathsf{B}(\lambda |\tau +\eta ) \\ 
     -\mathsf{C}(\lambda |\tau -\eta ) & \mathsf{A}(\lambda |\tau -\eta )
  \end{pmatrix}_{\! [0]}\,
  e^{i\mathsf{y}\eta\mathsf{S}}\,
  \frac{\theta (\tau +\eta \text{$\mathsf{S}$})}{\theta (\tau )}
  \\
  =\begin{pmatrix}
     \mathsf{D}^{\mathsf{(8V)}}(\lambda ) & -\mathsf{B}^{\mathsf{(8V)}}(\lambda )\\ 
    -\mathsf{C}^{\mathsf{(8V)}}(\lambda ) & \mathsf{A}^{\mathsf{(8V)}}(\lambda )
    \end{pmatrix}_{\! [0]}\,
    S_{q}(\tau )\, S_{0}(\lambda +\eta |\tau +\eta \mathsf{S}),
\end{multline}
where we have used the inversion formulas \eqref{Inv-8v-M} and \eqref{inv-mon} for the 8-vertex and dynamical 6-vertex monodromy matrices.
By means of the relation, for $\mathsf{x,y}\in\{0,1\}$,
\begin{align}
   S(\lambda|-\tau+\mathsf{x}\pi+\mathsf{y}\pi\omega)
   &=(-1)^\mathsf{x}\, i^{\mathsf{xy}}\, (\sigma^z)^\mathsf{x}\, (\sigma^x)^\mathsf{y}\, S(-\lambda|\tau)
   \nonumber\\
   &=(-1)^\mathsf{x}\, i^{\mathsf{xy}}\, (\sigma^z)^\mathsf{x}\, (\sigma^x)^\mathsf{y}\, S(\lambda|\tau)\, \sigma^x,\label{transf-S}
\end{align}
we can rewrite \eqref{v-IRF1} as
\begin{multline}\label{v-IRF2}
  S_{0}(\lambda +\eta |-\tau +\mathsf{x}\pi+\mathsf{y}\pi\omega)\, S_{q}(\tau -\eta \sigma _{0}^{z})\,
  \begin{pmatrix}
     -\mathsf{C}(\lambda |\tau -\eta ) & \mathsf{A}(\lambda |\tau -\eta ) \\ 
     \mathsf{D}(\lambda |\tau +\eta ) & -\mathsf{B}(\lambda |\tau +\eta )
  \end{pmatrix}_{\! [0]}\,
  e^{i\mathsf{y}\eta\mathsf{S}}\,
  \frac{\theta (\tau +\eta \text{$\mathsf{S}$})}{\theta (\tau )}
  \\
  =i^{\mathsf{x y}}\, (-1)^\mathsf{x}\,(\sigma_0^z)^\mathsf{x}\, (\sigma_0^x)^\mathsf{y}\,
  \begin{pmatrix}
     \mathsf{D}^{\mathsf{(8V)}}(\lambda ) & -\mathsf{B}^{\mathsf{(8V)}}(\lambda )\\ 
    -\mathsf{C}^{\mathsf{(8V)}}(\lambda ) & \mathsf{A}^{\mathsf{(8V)}}(\lambda )
    \end{pmatrix}_{\! [0]}\,
    S_{q}(\tau )\, S_{0}(\lambda +\eta |\tau +\eta \mathsf{S}).
\end{multline}

Let us now show that $S_{0}(\lambda +\eta |\tau+\eta \mathsf{S})$ is an invertible matrix in the auxiliary space when acting on any
state of $\mathbb{\bar{D}}_{\mathsf{(6VD)},\mathsf{N}}^{(0,\mathcal{R})}$.
All that we need to observe is that the following identity holds:
\begin{multline*}
   S_{0}(\lambda +\eta |\tau +\eta \mathsf{S})
   \left[ \big( \otimes _{n=1}^{\mathsf{N}} | n,h_{n}\rangle\big) \otimes | t_{0,\mathbf{h}}\rangle \right]  
   =e^{i\frac{\mathsf{y}}{2}(\mathsf{x}\pi+\mathsf{y}\pi\omega-t_{0,\mathbf{h}})}
   \left[ \big( \otimes _{n=1}^{\mathsf{N}} | n,h_{n}\rangle\big) \otimes | t_{0,\mathbf{h}}\rangle \right]  
        \\
      \times
     \begin{pmatrix}
    e^{-i\mathsf{y}\frac{\lambda+\eta}{2}}
    \theta _{2}(-\lambda -\eta -t_{0,\mathbf{h}}+\mathsf{x}\pi+\mathsf{y}\pi\omega|2\omega) 
    & e^{i\mathsf{y}\frac{\lambda+\eta}{2}}
    \theta _{2}(\lambda+\eta -t_{0,\mathbf{h}}+\mathsf{x}\pi+\mathsf{y}\pi\omega|2\omega)
     \\ 
    e^{-i\mathsf{y}\frac{\lambda+\eta}{2}}
    \theta _{3}(-\lambda -\eta -t_{0,\mathbf{h}}+\mathsf{x}\pi+\mathsf{y}\pi\omega|2\omega) 
    & e^{i\mathsf{y}\frac{\lambda+\eta}{2}}
    \theta _{3}(\lambda+\eta -t_{0,\mathbf{h}}+\mathsf{x}\pi+\mathsf{y}\pi\omega|2\omega)
   \end{pmatrix}_{\! [0]},
\end{multline*}
and that the last matrix has a non-zero determinant $\theta (\lambda +\eta )\,\theta (-t_{0,\mathbf{h}}+\mathsf{x}\pi+\mathsf{y}\pi\omega)$ for any $\mathbf{h}\in\{0,1\}^\mathsf{N}$ provided that $\lambda+\eta\notin\pi\mathbb{Z}+\pi\omega\mathbb{Z}$.

Multiplying both sides of \eqref{v-IRF2} from the right by $\left[ S_{0}(\lambda +\eta |\tau +\eta \mathsf{S})\right] ^{-1}$ and taking the trace on the auxiliary space $0$, we obtain
\begin{multline}
  \tr_0 \bigg\{ S_{0}(\lambda +\eta |-\tau +\mathsf{x}\pi+\mathsf{y}\pi\omega)\, S_{q}(\tau -\eta \sigma _{0}^{z})\,
  \begin{pmatrix}
     -\mathsf{C}(\lambda |\tau -\eta ) & \mathsf{A}(\lambda |\tau -\eta ) \\ 
     \mathsf{D}(\lambda |\tau +\eta ) & -\mathsf{B}(\lambda |\tau +\eta )
  \end{pmatrix}_{\! [0]}\,
  e^{i\mathsf{y}\eta\mathsf{S}}\\
  \times
  \frac{\theta (\tau +\eta \text{$\mathsf{S}$})}{\theta (\tau )}\, 
  \left[ S_{0}(\lambda +\eta |\tau +\eta \mathsf{S})\right] ^{-1}\bigg\}
  =(-1)^\mathsf{y}\, i^{\mathsf{xy}}\, \mathsf{T}^{\mathsf{(8V)}}_{\mathsf{(x,y)}}(\lambda)\, S_q(\tau).
\end{multline}
The claim \eqref{right-act-T}  follows by passing $S_{0}(\lambda +\eta |-\tau +\mathsf{x}\pi+\mathsf{y}\pi\omega)$ to the right in the trace and by recalling that $\mathbb{\bar{D}}_{\mathsf{(6VD)},\mathsf{N}}^{(0,\mathcal{R})}$ is an eigenspace of $\eta\mathsf{S}+2\tau$ associated with the eigenvalue $\mathsf{x}\pi+\mathsf{y}\pi\omega$.

\eqref{right-act-T2} is then obtained by applying $\mathbf{P}^\mathcal{R}$ on both members of \eqref{right-act-T}, and by using \eqref{P-S}, \eqref{actC1} and \eqref{actB1}.
\qed

Note that the relation between the action to the right of the 8-vertex and SOS transfer matrices of Proposition~\ref{th-act-T} can be easily extended to a relation between the action of more general matrix elements. 
Indeed, defining the coefficients $s_{i, j}^{\alpha,\beta}(\lambda|\tau)$ and $\tilde{s}_{i, j}^{\alpha,\beta}(\lambda|\tau)$, $i,j,\alpha,\beta\in\{+,-\}$ (which can easily be explicitly computed) by
\begin{align}
    &S(\lambda| t)^{-1}\, E^{ij}\, S(\lambda| t)
    =\!\!\!\!\sum_{\alpha,\beta\in\{+,-\} }\!\!\!\! s_{i, j}^{\alpha,\beta}(\lambda|\tau)\, E^{\alpha \beta},
    \\
    &S(\lambda| t)\, E^{ij}\, S(\lambda| t)^{-1}
    =\!\!\!\!\sum_{\alpha,\beta\in\{+,-\} }\!\!\!\! \tilde{s}_{i, j}^{\alpha,\beta}(\lambda|\tau)\, E^{\alpha \beta},
\end{align}
so that
\begin{equation}
   \sum_{\alpha,\beta\in\{+,-\} }\!\!\!\! s^{i, j}_{\alpha,\beta}(\lambda|\tau)\, \tilde{s}_{k, \ell}^{\alpha,\beta}(\lambda|\tau) 
   =  \sum_{\alpha,\beta\in\{+,-\} }\!\!\!\! \tilde{s}^{i, j}_{\alpha,\beta}(\lambda|\tau)\, {s}_{k, \ell}^{\alpha,\beta}(\lambda|\tau)=\delta_{i,k}\,\delta_{j,\ell},
\end{equation}
it is easy to show that the elements of the $\mathsf{(x,y)}$-twisted inverse monodromy matrix combined with the vertex-IRF transformation have the following action on the states $\ket{\mathbf{v}}$ of $\bar{\mathbb{D}}^{(0,\mathcal{R})}_{\mathsf{(6VD),N}}$:
\begin{equation}\label{act-el1}
   \big[ \mathsf{M}_\mathsf{(x,y)}^\mathsf{(8V)}(\lambda)^{-1}\big]_{kj}\, S_q(\tau)\, \ket{\mathbf{v}}
   = (-i)^\mathsf{xy} (-1)^\mathsf{x}
    \!\!\!\!\!\sum_{\alpha,\beta\in\{+,-\} }\!\!\!\!\!
    s_{j,k}^{\alpha,\beta}(\lambda|\tau)\, 
   S_q(\tau+\eta\beta) \big[ \bar{\mathsf{M}}(\lambda|\tau)^{-1}\big]_{\beta \alpha} \ket{\mathbf{v}},
\end{equation}
or equivalently
\begin{equation}\label{act-el2}
   S_q(\tau+\eta k)\,  \big[ \bar{\mathsf{M}}(\lambda|\tau)^{-1}\big]_{kj}\,  \ket{\mathbf{v}}
   = i^\mathsf{xy} (-1)^\mathsf{x}
     \!\!\!\!\! \sum_{\alpha,\beta\in\{+,-\} }\!\!\!\!\!
    \tilde{s}_{j,k}^{\alpha,\beta}(\lambda|\tau)\,
    \big[ \mathsf{M}_\mathsf{(x,y)}^\mathsf{(8V)}(\lambda)^{-1}\big]_{\beta \alpha}\, S_q(\tau)\, \ket{\mathbf{v}}.
\end{equation}
Note that these relations can be rewritten in terms of the matrix elements of $\bar{\mathsf{M}}(\lambda|\tau)$ and/or $\mathsf{M}_\mathsf{(x,y)}^\mathsf{(8V)}(\lambda)$ (instead of their inverse) by using \eqref{inv-mon} and/or \eqref{Inv-8v-M}.


\section{Spectral problem for the quasi-periodic 8-vertex transfer matrices}
\label{sec-diag-8V}

From Proposition~\ref{th-act-T} or Corollary~\ref{cor-rel-T} and the fact that $\mathbf{S}^{(0)}\,\mathbf{P}^{(0)}$ defines, for $\mathsf{(x,y)}\not=(0,0)$, an isomorphism from $\mathbb{\bar{D}}_{\mathsf{(6VD)},\mathsf{N}}^{(0)}$ to $\mathbb{V}_\mathsf{N}$, the spectral problem for the quasi-periodic 8-vertex transfer matrix \eqref{transfer} with $\mathsf{(x,y)}\not=(0,0)$ is completely equivalent to the spectral problem for the transfer matrix \eqref{6VD-antiT} of the antiperiodic  $\mathsf{(x,y)}$-dynamical 6-vertex model.
This property can be formulated as follows.

\begin{theorem}\label{th-relation-spectre}
Let $\mathsf{(x,y)}\not=(0,0)$.
If
\begin{equation}
   \ket{ \Psi_{\bar{\mathsf{t}}}^\mathsf{(6VD)} }\in\bar{\mathbb{D}}_\mathsf{(6VD),N}^{(0,\mathcal{R})},
   \qquad\text{respectively}\qquad
   \bra{  \Psi_{\bar{\mathsf{t}}}^\mathsf{(6VD)}  }\in\bar{\mathbb{D}}_\mathsf{(6VD),N}^{(0,\mathcal{L})},
\end{equation}
is a right (resp. left) eigenvector of the antiperiodic $\mathsf{(x,y)}$-dynamical 6-vertex transfer matrix $\overline{\mathcal{T}}(\lambda)$ \eqref{6VD-antiT} with eigenvalue $\bar{\mathsf{t}}(\lambda)$, then
\begin{equation}\label{eigen8V-6VD}
   \mathbf{S}^{(0)}\,\mathbf{P}^{(0)}\, 
   \ket{  \Psi_{\bar{\mathsf{t}}}^\mathsf{(6VD)}  }\in\mathbb{V}_\mathsf{N}^\mathcal{R},
   \qquad \text{resp.}\qquad
   \bra{  \Psi_{\bar{\mathsf{t}}}^\mathsf{(6VD)}  }\,\big[\mathbf{P}^{(0)}\big]^{-1}\, \big[\mathbf{S}^{(0)}\big]^{-1}\in\mathbb{V}_\mathsf{N}^\mathcal{L},
\end{equation}
is a right (resp. left) eigenvector of the quasi-periodic 8-vertex transfer matrix $\mathsf{T}_{\mathsf{(x,y)}}^{\mathsf{(8V)}}(\lambda )$ \eqref{transfer} with eigenvalue
\begin{equation}
   \mathsf{t}_{\mathsf{(x,y)}}^{\mathsf{(8V)}}(\lambda )\equiv (-1)^\mathsf{x}\, i^\mathsf{xy}\, \bar{\mathsf{t}}(\lambda),
\end{equation}
and conversely.
\end{theorem}

\subsection{Complete spectrum and eigenstate construction of the 8-vertex quasi-periodic transfer matrices}
\label{ssec-spectr-8V}

The complete description of the spectrum and eigenstates of the transfer matrices \eqref{transfer} hence follows directly from Theorem~\ref{th-relation-spectre} and from the complete description of the spectrum and eigenstates of the antiperiodic dynamical 6-vertex transfer matrix \eqref{6VD-antiT} which has been obtained in \cite{LevNT15} by means of Sklyanin's quantum Separation of Variable approach \cite{Skl90,Skl92} (see  Appendix~\ref{app-spectr-6VD} for a briery summary of the SOV study of \cite{LevNT15}). Notably, the 8-vertex transfer matrix eigenstates can be defined in a self-contained way from the image on $\mathbb{V}_\mathsf{N}$ of the SOV-basis of $\bar{\mathbb{D}}_{\mathsf{(6VD)},\mathsf{N}}^{(0)}$ of \cite{LevNT15} by the isomorphism $\mathbf{P}^{(0)}$.

Concretely let us define,  for each $\mathsf{N}$-tuple $\mathbf{h}\equiv (h_{1},\ldots,h_{\mathsf{N}})\in \{0,1\}^{\mathsf{N}}$, the following states in $\mathbb{V}_\mathsf{N}^\mathcal{L}$ and $\mathbb{V}_\mathsf{N}^\mathcal{R}$ respectively:
\begin{equation}\label{states-bar}
  \bra{ \underline{\mathbf{h}}}\equiv \bra{\mathbf{h}}\,\big[ \mathbf{P}^{(0)}\big]^{-1}
   = \bra{\mathbf{h}}\, t_{0,\mathbf{h}}\,\rangle,
   \qquad\text{resp.}\qquad
   \ket{\underline{\mathbf{h}}}\equiv \mathbf{P}^{(0)}\,\ket{\mathbf{h}}
   = \langle\, t_{0,\mathbf{h}}\,\ket{\mathbf{h}},
\end{equation}
where $\bra{\mathbf{h}}$ and $\ket{\mathbf{h}}$ stand for the states \eqref{D-left-eigenstates} and \eqref{D-right-eigenstates} respectively.
Then, under the condition \eqref{cond-inh} on the inhomogeneity parameters, the states \eqref{states-bar} define a basis of $\mathbb{V}_\mathsf{N}^\mathcal{L}$ and $\mathbb{V}_\mathsf{N}^\mathcal{R}$ respectively. 
Note that these states can equivalently be defined by multiple action of the operators \eqref{def-BC}
\begin{alignat}{2}
 &\bra{\underline{\mathbf{h}}}= 
    \bra{\underline{\mathbf{0}}} \prod_{n=1}^{\mathsf{N}}
 \left( \frac{\mathsf{C^{(6VD)}}(\xi_{n} )}{\text{\textsc{d}}(\xi _{n}-\eta )}\right)^{\! h_{n}},
\qquad
 &&\ket{\underline{\mathbf{h}} } 
 = \prod_{n=1}^{\mathsf{N}}
 \left( \frac{\mathsf{C^{(6VD)}}(\xi _{n}-\eta  )}{\text{\textsc{d}}(\xi _{n}-\eta )}\right) ^{\! (1-h_{n})}
 \ket{\underline{\mathbf{1}} } ,
%
\intertext{on the following reference states}
%
  &\bra{\underline{\mathbf{0}} }
\equiv \frac{1}{\mathsc{n}}
\big(\otimes _{n=1}^{\mathsf{N}}\langle n,h_{n}=0| \big),
\qquad
 && \ket{\underline{\mathbf{1} }}
 \equiv \frac{1}{\mathsc{n}}
 \big( \otimes _{n=1}^{\mathsf{N}}|n,h_{n}=1\rangle \big).
\end{alignat}
The action of the operators $\mathsf{B^{(6VD)}}(\lambda)$ and $\mathsf{C^{(6VD)}}(\lambda)$ on this basis can immediately be deduced from the action \eqref{C-SOV_D-left}-\eqref{B-SOV_D-right} of the operators $\mathcal{B}(\lambda)$ and $\mathcal{C}(\lambda)$ on the corresponding basis of $\mathbb{\bar{D}}_{\mathsf{(6VD)},\mathsf{N}}^{(0,\mathcal{R/L})}$.
Moreover,
\begin{equation}\label{sc-prod}
  \moy{  \underline{\mathbf{h}}\, |\,\underline{ \mathbf{k} }}
   = \delta_{\mathbf{h},\mathbf{k}}\,
   \frac{e^{-i\mathsf{y}\eta \sum_{j=1}^\mathsf{N}h_j}}
           {\det_{\mathsf{N}}\big[ \Theta^{(\mathbf{h})} \big]},
\end{equation}
where $\Theta^{(\mathbf{h})}$ is the matrix \eqref{mat-Theta}. Hence these states define a decomposition of the identity on $\mathbb{V}_{\mathsf{N}}$: 
\begin{equation}
     \mathbb{I}\equiv \sum_{\mathbf{h}\in\{0,1\}^\mathsf{N}}
     \det_{\mathsf{N}}\big[ \Theta^{(\mathbf{h})} \big]\
    \ket{\underline{\mathbf{h}}}\,\bra{\underline{\mathbf{h}}}.  \label{Decomp-Id}
\end{equation}

\begin{theorem}
\label{th-T8V-diag}
Let $\mathsf{(x,y)}\not=(0,0)$.
For any fixed $\mathsf{N}$-tuple of
inhomogeneities $(\xi _{1},\ldots,\xi _{\mathsf{N}})\in \mathbb{C}^{\mathsf{N}}$ satisfying \eqref{cond-inh0}-\eqref{cond-inh}, the spectrum $\Sigma_{\mathsf{(x,y)}}^\mathsf{(8V)}$ of the quasi-periodic 8-vertex transfer matrix $\mathsf{T}^{\mathsf{(8V)}}_\mathsf{(x,y)}(\lambda )$ \eqref{transfer} is simple and coincides with the set of functions of the form
\begin{equation}
   \mathsf{t}(\lambda )\equiv 
   \sum_{a=1}^{\mathsf{N}} e^{i\mathsf{y}(\xi_a-\lambda)}\,
\frac{\theta(t_{0,\mathbf{0}}-\lambda +\xi _{a})}{\theta (t_{0,\mathbf{0}})}
\prod_{b\neq a}\frac{\theta (\lambda -\xi _{b})}{\theta (\xi _{a}-\xi _{b})}\,
 \mathsf{t}(\xi_a) ,
 \quad \big({\mathsf{t}}(\xi _1),\ldots, {\mathsf{t}}(\xi _{\mathsf{N}})\big)\in\mathbb{C}^\mathsf{N},
\label{t-8Vxy}
\end{equation}
which satisfy the discrete system of equations \eqref{eq-quadr-8V}.

The right $\mathsf{T}^{\mathsf{(8V)}}_\mathsf{(x,y)}(\lambda )$-eigenstate $\ket{ \Psi_\mathsf{t}^\mathsf{(8V)} }\in\mathbb{V}_\mathsf{N}^{\mathcal{R}}$ and the left $\mathsf{T}^{\mathsf{(8V)}}_\mathsf{(x,y)}(\lambda )$-eigenstate $\bra{ \Psi_\mathsf{t}^\mathsf{(8V)} }\in\mathbb{V}_\mathsf{N}^{\mathcal{R}}$ associated with the eigenvalue $\mathsf{t}(\lambda )\in \Sigma_\mathsf{(x,y)}^\mathsf{(8V)}$ are respectively given by
\begin{align}
& \ket{ \Psi_\mathsf{t}^\mathsf{(8V)}}
=\sum_{\mathbf{h}\in\{0,1\}^\mathsf{N}}
\prod_{a=1}^{\mathsf{N}}\bigg[ e^{i\mathsf{y}\eta h_a}\, \bigg( \frac{\mathsc{a}_\mathsf{x,y}(\xi_a)}{\mathsc{d}(\xi_a-\eta)}\bigg)^{\! h_a}\,  \mathsf{q}_{{\mathsf{t}},a}^{(h_a)}\bigg]\,
\det_{\mathsf{N}}\big[\Theta^{(\mathbf{h}) }\big]\ 
\mathbf{S}^{(0)}\ket{\underline{\mathbf{h}}} , 
 \label{eigenR-XYZ}\\
&  \bra{ \Psi_\mathsf{t}^\mathsf{(8V)} }
  = \sum_{\mathbf{h}\in\{0,1\}^\mathsf{N}}
\prod_{a=1}^{\mathsf{N}}\left[ e^{i\mathsf{y}\eta h_a}\,  \mathsf{q}_{{\mathsf{t}},a}^{(h_a)}\right]\,
\det_{\mathsf{N}}\big[\Theta^{(\mathbf{h})}\big] \ 
\bra{\underline{ \mathbf{h}} }  \big[\mathbf{S}^{(0)}\big]^{-1}, 
 \label{eigenL-XYZ}
\end{align}
where $\mathsc{a}_\mathsf{x,y}(\lambda)=(-1)^\mathsf{x+y+xy}\mathsc{a}(\lambda)$, and where the coefficients $\mathsf{q}_{\bar{\mathsf{t}},a}^{(h_a)}$ are (up to an overall normalization) characterized by
\begin{equation}
\frac{ \mathsf{q}_{{\mathsf{t}},a}^{(1)} }{  \mathsf{q}_{{\mathsf{t}},a}^{(0)}  }
=(-1)^\mathsf{x}\, i^\mathsf{xy}\,\frac{\mathsc{d}(\xi_a-\eta)}{ \mathsf{t}(\xi_a-\eta)}
=(-1)^{\mathsf{y}}\, i^\mathsf{xy}\, \frac{ \mathsf{t}(\xi_a)}{\mathsc{a}(\xi_a)} .
\label{t-q-XYZ}
\end{equation}
\end{theorem}

It is interesting to remark that the 8-vertex eigenstates \eqref{eigenR-XYZ} and \eqref{eigenL-XYZ} have a complete separated form on the basis of $\mathbb{V}_\mathsf{N}^\mathcal{R}$ given by the states $\mathbf{S}^{(0)}\ket{\underline{\mathbf{h}}}$, $\mathbf{h}\in\{0,1\}^\mathsf{N}$, and on the basis of $\mathbb{V}_\mathsf{N}^\mathcal{L}$ given by the states $\bra{\underline{ \mathbf{h}} }  \big[\mathbf{S}^{(0)}\big]^{-1}$,  $\mathbf{h}\in\{0,1\}^\mathsf{N}$, respectively.
In other words, they belong to the class of states (that we shall call {\em separate states}) which can be written under the following form,
\begin{align}
 & \ket{ \alpha}
=\sum_{\mathbf{h}\in\{0,1\}^\mathsf{N}}
\prod_{a=1}^{\mathsf{N}}\bigg[ e^{i\mathsf{y}\eta h_a}\, \bigg( \frac{\mathsc{a}_\mathsf{x,y}(\xi_a)}{\mathsc{d}(\xi_a-\eta)}\bigg)^{\! h_a}\,  \alpha(\xi_a-\eta h_a)\bigg]\,
\det_{\mathsf{N}}\big[\Theta^{(\mathbf{h}) }\big]\ 
\mathbf{S}^{(0)}\ket{\underline{\mathbf{h}}} , 
 \label{state-alpha}\\
&  \bra{ \beta }
  = \sum_{\mathbf{h}\in\{0,1\}^\mathsf{N}}
\prod_{a=1}^{\mathsf{N}}\left[ e^{i\mathsf{y}\eta h_a}\,  \beta(\xi_a-\eta h_a)\right]\,
\det_{\mathsf{N}}\big[\Theta^{(\mathbf{h})}\big] \ 
\bra{\underline{ \mathbf{h}} }  \big[\mathbf{S}^{(0)}\big]^{-1}, 
\label{state-beta}
\end{align}
in terms of any function $\alpha$ or $\beta$  on the discrete set $\Xi\equiv \{\xi_j,\xi_j-\eta\}_{1\le j\le \mathsf{N}}$.
It follows from \eqref{sc-prod} that the state $\bra{\underline{ \mathbf{h}} }  \big[\mathbf{S}^{(0)}\big]^{-1}$ is proportional to the dual state of $\mathbf{S}^{(0)}\ket{\underline{\mathbf{h}}}$. Hence the scalar product of any separate states of the form \eqref{state-alpha}-\eqref{state-beta} (and in particular of two eigenstates  $\ket{ \Psi_\mathsf{t}^\mathsf{(8V)}}$ and $\bra{ \Psi_{\mathsf{t}'}^\mathsf{(8V)}}$ if one sets $\alpha(\xi_a-\eta h_a)=\mathsf{q}_{{\mathsf{t}},a}^{(h_a)}$ and $\beta(\xi_a-\eta h_a)=\mathsf{q}_{{\mathsf{t}'},a}^{(h_a)}$), can be expressed as a simple determinant, in a form which is quite general for the models solved by SOV (see for instance \cite{Nic13,KitMNT15}):
\begin{equation}\label{sp-gen}
   \moy{\beta\, |\, \alpha}=\det_\mathsf{1\le j,k\le N}\left[\sum_{h=0}^1\bigg(e^{i\mathsf{y}\eta}\, \frac{\mathsc{a}_\mathsf{x,y}(\xi_j)}{\mathsc{d}(\xi_j-\eta)}\bigg)^{\! h} \alpha(\xi_j-h\eta)\, \beta(\xi_j-h\eta)\,\vartheta_{k-1}(\xi_j-h\eta-\bar{\xi}_0)\right],
\end{equation}
in terms of the functions $\vartheta_k(\lambda)$ \eqref{def-theta_j} and the constant $\bar{\xi}_0$ \eqref{mat-Theta}.
Note that, contrary to what happens in the periodic case with $\mathsf{N}$ even for which the existence of a compact determinant representation for the scalar products of the 8-vertex Bethe states is still an open problem, the effect of the vertex-IRF transformation is here completely trivial. In fact, the expression~\eqref{sp-gen} just coincides with the scalar product of two arbitrary separate states in the antiperiodic dynamical 6-vertex model \cite{LevNT15}.

\subsection{Characterization of the spectrum and eigenstates through the solutions of a functional $T$-$Q$ equation}
\label{sec-eq-Baxter}

Theorem~\ref{th-T8V-diag} provides a complete description of the quasi-periodic 8-vertex transfer matrix spectrum and eigenstates, which is however not so convenient for the consideration of the thermodynamic and even the homogeneous limits of the model. To this aim it would be desirable to reformulate this characterization in terms of Bethe-type equations i.e., in terms of some particular classes of solutions of a functional $T$-$Q$ equation of Baxter's type \cite{Bax82L}, as it has already been done in the context of several other models solved by SOV \cite{NicT10,Nic10a,GroN12,NicT15,KitMNT15}.

In fact, this problem has already been considered in \cite{LevNT15} in the case of the antiperiodic dynamical 6-vertex model. There we have discussed  the existence of two different possible reformulations.

On the one hand, it has been proven that the SOV discrete characterization of the antiperiodic dynamical 6-vertex transfer matrix spectrum and eigenstates could be equivalently reformulated in terms of a particular class of solutions of some functional $T$-$Q$ equation with an extra inhomogeneous term (see Appendix B of \cite{LevNT15}). This reformulation can of course be translated to the quasi-periodic 8-vertex case\footnote{Note that in the case of the 8-vertex model with periodic or open boundary conditions similar types of inhomogeneous $T$-$Q$ equations had been previously proposed in \cite{CaoCYSW14} to describe the spectrum only, however without proof of completeness.}, and provides an equivalent complete description of the transfer matrix spectrum and eigenstates, the former in terms of solutions of Bethe-type equations with an extra inhomogeneous term,  and the latter in terms of the multiple action of some operator evaluated at the Bethe roots on a convenient pseudo-vacuum state (see Appendix~\ref{app-inhom} for more details). This reformulation which, except from the extra terms in the Bethe equations, presents many similarities with ABA (in particular from the way the eigenstates can be constructed), allows  for an easier consideration of the homogeneous limit of the model. However, its efficiency for the consideration of the thermodynamic limit is still not so clear due to the difficulties arising from the presence of the inhomogeneous term in the Bethe equations.

On the other hand, we have also studied in \cite{LevNT15} the possibility to characterize, in a probably more efficient way for the consideration of the thermodynamic limit, the antiperiodic dynamical 6-vertex spectrum and eigenstates in terms of solutions of the usual $T$-$Q$ equation (i.e. without extra inhomogeneous term), hence leading to a reformulation in terms of solutions of usual Bethe-type equations. This problem is in principle slightly more delicate, since we have to identify the functional form of the solutions and to show the completeness of this description. An ansatz has been proposed in \cite{LevNT15} concerning the functional form of the solutions, i.e. the form of the Bethe equations. The completeness of this ansatz has been proven in the case of a model with an even number of sites. We can therefore use this result to formulate directly the following theorem:

\begin{theorem}\label{th-hom-eq}
Let $\mathsf{(x,y)}\not= (0,0)$ and let us suppose that the inhomogeneity parameters of the model satisfy \eqref{cond-inh0}-\eqref{cond-inh}.
Then we have the following properties:
\begin{enumerate}
\item Let $\mathsf{t}(\lambda)$ be an entire function such that
\begin{enumerate}
\item\label{cond-a} there exists a function $Q(\lambda)$ such that $\mathsf{t}(\lambda)$ and $Q(\lambda)$ satisfy the functional equation
\begin{equation}\label{hom-eq}
   \mathsf{t}(\lambda)\, Q(\lambda) = (-1)^\mathsf{y}\,(-i)^\mathsf{xy}\,\mathsc{a}(\lambda)\, Q(\lambda-\eta)+(-1)^\mathsf{x}\, i^\mathsf{xy}\,\mathsc{d}(\lambda)\, Q(\lambda+\eta),
\end{equation}
\item this function $Q(\lambda)$ is such that, for each $n\in\{1,\ldots,\mathsf{N}\}$, there exists $(\alpha_n,\beta_n)\in\{0,1\}^2$ such that $\big( Q(\xi_n+\alpha_n\pi+\beta_n\pi\omega), Q(\xi_n+\alpha_n\pi+\beta_n\pi\omega-\eta)\big)\not=(0,0)$,
\item $\mathsf{t}(\lambda)$ satisfies the quasi-periodicity properties \eqref{periodt-1}-\eqref{periodt-2}.
\end{enumerate}
Then $\mathsf{t}(\lambda)$ is an eigenvalue of the $\mathsf{(x,y)}$-twisted 8-vertex transfer matrix \eqref{transfer}. The corresponding one-dimensional eigenspace is generated by the following vectors, which are proportional to each others:
\begin{multline}
  \ket{\Psi_\mathsf{t}^{(\alpha,\beta)} }
  =\sum_{\mathbf{h}\in\{0,1\}^\mathsf{N}}
\prod_{a=1}^{\mathsf{N}}\left[ 
\left( \frac{e^{i\mathsf{y}\eta}\mathsc{a}_\mathsf{x,y}(\xi_a)}{\mathsc{d}(\xi_a-\eta)}\right)^{\! h_a} 
e^{i [ \alpha_a\frac{\pi\mathsf{y}}{\eta}+\beta_a(\mathsf{N}+\frac{\pi\mathsf{x}}{\eta}) ](\xi_a-\eta h_a)}
\right]
\\
\times\prod_{a=1}^\mathsf{N} Q(\xi_a-\eta h_a+\alpha_a\pi+\beta_a\pi\omega)\
\det_{\mathsf{N}}\big[\Theta^{(\mathbf{h}) }\big]\ 
\mathbf{S}^{(0)}\ket{\underline{\mathbf{h}}} .
\end{multline}
\item Let $\mathsf{t}(\lambda )$ be an eigenvalue of the $\mathsf{(x,y)}$-twisted 8-vertex transfer matrix \eqref{transfer}. 
Then, if $\mathsf{N}$ is even, there exists a unique function $Q(\lambda )$ of the form
\begin{equation}\label{Q-X}
  Q(\lambda)=\prod_{j=1}^\mathsf{N}\theta_{\mathsf{X}}(\lambda-\lambda_j),
\end{equation}
for some set of roots $\lambda_1,\ldots,\lambda_\mathsf{N}\in\mathbb{C}$, such that $\mathsf{t}(\lambda )$ and $Q(\lambda )$ satisfy the homogeneous functional equation \eqref{hom-eq}.
This function $Q(\lambda)$ is such that, for each $n\in\{1,\ldots,\mathsf{N}\}$, $\big( Q(\xi_n-\eta),Q(\xi_n-\eta+\pi),Q(\xi_n-\eta+\pi\omega)\big)\not=(0,0,0)$.
In \eqref{Q-X}, the notation  $\theta_\mathsf{X}(\lambda)$ stands for the function 
\begin{equation}\label{theta-X}
  \theta_\mathsf{X}(\lambda)
  =\begin{cases}
     \theta_1\big(\frac{\lambda}{2}\,\big|\,\frac{\omega}{2}\big) 
     &\text{if }\ \mathsf{(x,y)}=(0,1),  
     \\
     \theta_1(\lambda|2\omega) 
     &\text{if }\ \mathsf{(x,y)}=(1,0),
     \\
     e^{i\frac{\lambda}{2}}\,\theta_1\Big(\frac{\lambda}{2}\,\Big|\,\omega\Big)\, \theta_1\Big(\frac{\lambda+\pi+\pi\omega}{2}\,\Big|\,\omega\Big) 
      &\text{if }\ \mathsf{(x,y)}=(1,1).
   \end{cases}
\end{equation}

\end{enumerate}
\end{theorem}

Hence Theorem~\ref{th-hom-eq} provides, at least in the even $\mathsf{N}$ case, an alternative description of the spectrum and eigenstates of $\mathsf{T}^{\mathsf{(8V)}}_\mathsf{(x,y)}(\lambda )$ in terms of solutions of a system of Bethe equations.
These Bethe equations are written as the entireness condition for a function  $\mathsf{t}(\lambda)$ defined by the relation \eqref{hom-eq} in terms of a function $Q(\lambda)$ of the form \eqref{Q-X}-\eqref{theta-X}, i.e. as a system of equations for the roots $\lambda_j$, $j=1,\ldots,\mathsf{N}$, of the function $Q(\lambda)$  \eqref{Q-X}-\eqref{theta-X}.

Moreover, it is still possible to rewrite the corresponding eigenstates in a form more similar to what we have in the ABA framework, i.e. by multiple action, on a given pseudo-vacuum state, of a product of operators evaluated at the Bethe roots. Indeed, it has first been  shown in \cite{DerKM03} in the case of the non-compact quantum $SL(2,\mathbb{R}$) spin chain that the SOV representation of the transfer matrix eigenstates associated with a polynomial $Q$-function admits a rewriting in the ABA form. The arguments of \cite{DerKM03} can easily be adapted to other quantum integrable models solved by SOV provided the corresponding $Q$-functions still admit some (possibly model-dependent) generalized polynomial form. This is the case here, and we can define some adequate pseudo-vacuum states and represent the eigenstates by a product of diagonal operators in the SOV basis evaluated at the Bethe roots $\lambda_j$, similarly as what has been done in  \cite{LevNT15} for the antiperiodic dynamical 6-vertex model. However, a crucial difference with usual ABA (and with respect to what is obtained when considering the simpler XXX model \cite{KitMNT15}) is that the operators which we use here to generate the eigenstates are not directly some of the generators of the Yang-Baxter algebra.
This is due to the fact that the $Q$-function has not the same functional form as the usual functions of the model.

Concretely,
for each $\mathsf{N}$-tuple $\boldsymbol{\beta}=(\beta_1,\ldots,\beta_\mathsf{N})\in\{0,1\}^\mathsf{N}$, we define an  operator $D_{\boldsymbol{\beta}}(\lambda)$ by its diagonal action in the SOV basis \eqref{states-bar} as%
\begin{equation}\label{def-Dbeta}
     D_{\boldsymbol{\beta}}(\lambda)\,  \ket{\underline{\mathbf{h}}}
     =\prod_{n=1}^\mathsf{N}\!
     \left\{ 
     \Big[c_\mathsf{X}^{\vphantom{-1}}\, e^{\frac{i\pi(\mathsf{x+y-xy})h_n}{\mathsf{N}}+i\delta_\mathsf{y=0}(\lambda-\xi_n^{(h_n)})} \Big]^{\beta_n}\,
     \theta_\mathsf{X}\big(\lambda-\xi_n^{(h_n)}+\beta_n\pi_\mathsf{X}^{\vphantom{1}}\big) \right\}
       \ket{\underline{\mathbf{h}}},
\end{equation}
where
\begin{align}
  &c_\mathsf{X}^{-1}=
  \begin{cases}
  \theta_4(0|\omega) &\text{if } \mathsf{x}=0,\\
  \frac{i}{2}\, e^{-i\frac{\pi\omega}{2}}\, \theta_2(0|\omega) &\text{if } \mathsf{y}=0,\\
  \frac{1}{2}\, e^{-i\frac{\pi\omega}{2}}\, \theta_2(0|\omega)\,\theta_3(0|\omega)\,\theta_4(0|\omega)\quad&\text{if } \mathsf{x=y},
  \end{cases}
  \\
  &\pi_\mathsf{X}^{\vphantom{1}}=(1-\delta_{\mathsf{y}=0})\,\pi+\delta_{\mathsf{y}=0}\, \pi\omega.
\end{align}
Let us remark that, for each $\beta\in\{0,1\}^\mathsf{N}$, the operator
\begin{equation}
\bar{D}(\lambda )=e^{\frac{i\pi (\mathsf{x+y-xy})(\mathsf{S}-\mathsf{N})}{2\mathsf{N}}}\,
D_{\boldsymbol{\beta}}(\lambda )\,D_{\boldsymbol{1-\beta}}(\lambda ),
\end{equation}
where $\boldsymbol{1-\beta}$ stands for the $\mathsf{N}$-tuple $(1-\beta_{1},\ldots ,1-\beta _{\mathsf{N}})$, does not depend on $\boldsymbol{\beta}$ and has the following simple diagonal form on the SOV basis \eqref{states-bar}:
\begin{equation} \label{bar-D}
\bar{D}(\lambda )\,\ket{\underline{\mathbf{h}}}
=\prod_{n=1}^{\mathsf{N}}\theta (\lambda -\xi _{n}^{(h_{n})})\,\ket{\underline{\mathbf{h}}}.
\end{equation}
Let us also define the following right $\ket{\underline{\Omega}}$ and left $\bra{ \underline{\Omega} }$ pseudo-vacuum states as the simplest separate states of the form \eqref{state-alpha} and \eqref{state-beta} associated with the function with constant value 1 on the discrete set $\Xi\equiv \{\xi_j,\xi_j-\eta\}_{1\le j\le \mathsf{N}}$:
\begin{align}
& \ket{\underline{\Omega}}  
   =\sum_{\mathbf{h}\in \{0,1\}^{\mathsf{N}}}
     \prod_{a=1}^{\mathsf{N}}\left( 
     \frac{e^{i\mathsf{y}\eta}\, \mathsc{a}_{\mathsf{x,y}}(\xi _{a})}{\mathsc{d}(\xi _{a}-\eta )}\right) ^{\!h_{a}}
     \det_{\mathsf{N}}\big[\Theta ^{(0,\mathbf{h})}\big]\ 
     \mathbf{S}^{(0)}\ket{\underline{\mathbf{h}}} ,  \label{ref1} \\
& \bra{ \underline{\Omega} }
   =\sum_{\mathbf{h}\in \{0,1\}^{\mathsf{N}}}\prod_{a=1}^{\mathsf{N}} e^{i\mathsf{y}h_a \eta }\
     \det_{\mathsf{N}}\big[\Theta ^{(0,\mathbf{h})}\big]\,
     \bra{\underline{ \mathbf{h}} }  \big[\mathbf{S}^{(0)}\big]^{-1}.  \label{ref2}
\end{align}
Then one can use our ansatz of \cite{LevNT15} to formulate the following proposition which, when $\mathsf{N}$ is even, is just a corollary of Theorem~\ref{th-hom-eq}:

\begin{proposition}\label{Cor-Eigen-Bethe}
Let the condition \eqref{cond-inh} be satisfied and let us denote with $\Sigma_\mathrm{BAE}$ the set of different (up to the real quasi-period of $\theta_\mathsf{X}$) Bethe roots $\Lambda=\{\lambda_1,\ldots,\lambda_\mathsf{N}\}$ defined by the requirements
\begin{enumerate}
  \item there exists $h\in\{0,1\}$ such that the following function 
\begin{equation}\label{Bethe-fct}
    \mathsf{t}(\la)\equiv(-1)^\mathsf{y}\,(-i)^\mathsf{xy}\,e^{ih(1-\mathsf{y})\eta}\,\mathsc{a}(\lambda)\, \frac{{Q}(\lambda-\eta)}{{Q}(\lambda)}
    +(-1)^\mathsf{x}\, i^\mathsf{xy}\,e^{-ih(1-\mathsf{y})\eta}\,\mathsc{d}(\lambda)\, \frac{{Q}(\lambda+\eta)}{{Q}(\lambda)},
\end{equation}
is entire and satisfies the quasi-periodicity properties \eqref{periodt-1}-\eqref{periodt-2},
  \item there exists a $\mathsf{N}$-tuple $\boldsymbol{\beta}\in\{0,1\}^\mathsf{N}$ such that ${Q}(\xi_n+\beta_n\pi_\mathsf{X}^{\vphantom{1}})\not= 0$,
  $\forall n\in\{1,\ldots,\mathsf{N}\}$,
\end{enumerate}
where 
$Q(\lambda)$ is defined in terms of $\Lambda$ by \eqref{Q-X}.
Then for any $\Lambda\in\Sigma_\mathrm{BAE}$ the entire function \eqref{Bethe-fct} belong to the spectrum $\Sigma_{\mathsf{(x,y)}}^\mathsf{(8V)}$ of the $\mathsf{(x,y)}$-twisted transfer matrix $\mathsf{T}^{\mathsf{(8V)}}_\mathsf{(x,y)}(\lambda )$ \eqref{transfer} and the corresponding one-dimensional right and left  eigenspaces   are the one-dimensional subspaces of $\mathbb{V}_\mathsf{N}^{\mathcal{R/L}}$  spanned by all vectors of the type
\begin{equation}\label{Bethe-eigen}
  \ket{\Psi_{\Lambda, \boldsymbol{\beta}} }
    =\prod_{j=1}^{\mathsf{N}} D_{\boldsymbol{\beta}}(\lambda_j)\,
       \ket{\underline{\Omega}} ,
       \qquad\text{respectively}\quad
  \bra{ \Psi_{\Lambda, \boldsymbol{\beta}  } }
    =\bra{ \underline{\Omega} }
      \prod_{j=1}^{\mathsf{N}} D_{\boldsymbol{\beta}}(\lambda_j) ,
\end{equation}
for any $\mathsf{N}$-tuple $\boldsymbol{\beta}\in\{0,1\}^\mathsf{N}$ satisfying 2.
In \eqref{Bethe-eigen}, the operators $D_{\boldsymbol{\beta}}(\lambda)$ are defined as in \eqref{def-Dbeta}, and the pseudo-vacuum states $\ket{\underline{\Omega}}$ and $\bra{ \underline{\Omega} }$ are the simplest separate states \eqref{ref1} and \eqref{ref2}.
If moreover $\mathsf{N}$ is even, then we can fix $h=0$, and by \eqref{Bethe-fct} and \eqref{Q-X} is defined a one-to-one correspondence between the sets $\Sigma_\mathrm{BAE}$ and $\Sigma_{\mathsf{(x,y)}}^\mathsf{(8V)}$, i.e. the Bethe ansatz equations are complete.
\end{proposition}


\section{The periodic case with an odd number of sites}
\label{sec-per}

The previous study does not directly apply to the periodic case $\mathsf{(x,y)}=(0,0)$ with $\mathsf{N}$ odd. 
In that case, the relations~\eqref{right-act-T}-\eqref{right-act-T2} still hold but, since $\mathbf{S}^{(0)}$ has a non-zero kernel, these relations are a priori not sufficient to completely determine the 8-vertex transfer matrix spectrum and eigenstates in terms of the dynamical 6-vertex ones. We shall see here that it is nevertheless possible to completely characterize the periodic 8-vertex transfer matrix spectrum and eigenstates.

The idea is to define another endomorphism $\hat{\mathbf{S}}^{(0)}$ of $\mathbb{V}_\mathsf{N}$ as
\begin{equation}\label{def-Sbar}
   \hat{\mathbf{S}}^{(0)}=\mathbf{S}^{(0)}\, \Gamma_z,\qquad
   \text{with}\quad \Gamma_z=\otimes_{n=1}^\mathsf{N}\sigma_n^z.
\end{equation}
It is indeed easy to see that, since both $R$-matrices \eqref{R-8V} and \eqref{R-6VD} commute with $\sigma^z\otimes\sigma^z$,  the periodic 8-vertex transfer matrix and the antiperiodic dynamical 6-vertex transfer matrix respectively  commutes and anti-commutes with $\Gamma_z$:
\begin{align}
  & \Gamma_z\, \mathsf{T}^{\mathsf{(8V)}}_\mathsf{(0,0)}(\lambda )\, \Gamma_z=\mathsf{T}^{\mathsf{(8V)}}_\mathsf{(0,0)}(\lambda ),
   \\
   &\Gamma_z\, \overline{\mathcal{T}}(\lambda )\, \Gamma_z=- \overline{\mathcal{T}}(\lambda ).
   \label{z-Tbar}
\end{align}
This implies the following additional relation, which can be deduced from \eqref{right-act-T2} for any state $\ket{\mathbf{v}}$ of $\mathbb{\bar{D}}_{\mathsf{(6VD)},\mathsf{N}}^{(0,\mathcal{R})}$:
\begin{equation}\label{right-act-T3}
   \mathsf{T}_{\mathsf{(0,0)}}^{\mathsf{(8V)}}(\lambda )\, \hat{\mathbf{S}}^{(0)}\,\mathbf{P}^{(0)}\, \ket{\mathbf{v}}
  =- \hat{\mathbf{S}}^{(0)}\,\mathbf{P}^{(0)}\,
  \overline{\mathcal{T}}(\lambda)\, \ket{\mathbf{v}}.
\end{equation}
In the following we clarify our interest in the introduction of this second gauge transformation. To this aim, let us start by some preliminary considerations about the spectrum of $\overline{\mathcal{T}}(\lambda)$ and of $\mathsf{T}_{\mathsf{(0,0)}}^{\mathsf{(8V)}}(\lambda )$.

\begin{lemma}\label{lem-sp-6VD}
The spectrum $\Sigma^\mathsf{(6VD)}$ of the antiperiodic dynamical 6-vertex transfer matrix $\overline{\mathcal{T}}(\lambda )$ in $\mathbb{\bar{D}}_{\mathsf{(6VD)},\mathsf{N}}^{(0)}$ can be decomposed as the union
\begin{equation}
    \Sigma^\mathsf{(6VD)}=\Sigma^\mathsf{(6VD)}_+\cup\Sigma^\mathsf{(6VD)}_-
\end{equation}
of two disjoint subsets $\Sigma^\mathsf{(6VD)}_+$ and $\Sigma^\mathsf{(6VD)}_-$ of the same cardinality $2^{\mathsf{N}-1}$, where
\begin{equation}\label{Sigma+}
   \Sigma^\mathsf{(6VD)}_+
   =\left\{ \bar{\mathsf{t}}_+(\lambda )\in\Sigma^\mathsf{(6VD)}\ \Bigg|\, \prod_{n=1}^\mathsf{N}\bar{\mathsf{t}}_+(\xi_n )=\prod_{n=1}^\mathsf{N}\mathsc{a}(\xi_n)\right\}
\end{equation}
and
\begin{equation}\label{Sigma-}
  \Sigma^\mathsf{(6VD)}_-
  =\left\{ \bar{\mathsf{t}}_-(\lambda )=-\bar{\mathsf{t}}_+(\lambda )\
  \Big|\ \bar{\mathsf{t}}_+(\lambda )\in\Sigma^\mathsf{(6VD)}_+\right\}.
\end{equation}
Let $\ket{ \Psi_{\bar{\mathsf{t}}_+}^\mathsf{(6VD)} }$ be the unique (up to normalization) $\overline{\mathcal{T}}(\lambda )$-eigenvector  with eigenvalue $\bar{\mathsf{t}}_+(\lambda )\in\Sigma^\mathsf{(6VD)}_+$, then $\Gamma_z\, \ket{ \Psi_{\bar{\mathsf{t}}_+} }$ is the unique (up to normalization) $\overline{\mathcal{T}}(\lambda )$-eigenvector  with eigenvalue $\bar{\mathsf{t}}_-(\lambda )=-\bar{\mathsf{t}}_+(\lambda )\in\Sigma^\mathsf{(6VD)}_-$.
\end{lemma}

In the following, we shall also denote by $\mathbf{\Sigma}_+^\mathsf{(6VD)}$ the subspace of   $\mathbb{\bar{D}}_{\mathsf{(6VD)},\mathsf{N}}^{(0)}$ spanned by all the $\overline{\mathcal{T}}(\lambda )$-eigenvectors with eigenvalue $\bar{\mathsf{t}}_+(\lambda )\in\Sigma^\mathsf{(6VD)}_+$, and  by $\mathbf{\Sigma}_-^\mathsf{(6VD)}$ the subspace of   $\mathbb{\bar{D}}_{\mathsf{(6VD)},\mathsf{N}}^{(0)}$ spanned by all the $\overline{\mathcal{T}}(\lambda )$-eigenvectors with eigenvalue $\bar{\mathsf{t}}_-(\lambda )\in\Sigma^\mathsf{(6VD)}_-$.
It follows from the previous considerations that $\mathbb{\bar{D}}_{\mathsf{(6VD)},\mathsf{N}}^{(0)}=\mathbf{\Sigma}_+^\mathsf{(6VD)}\oplus \mathbf{\Sigma}_-^\mathsf{(6VD)}$.

\begin{proof}
It is obvious, from the characterization \eqref{set-t}-\eqref{dis-sys}  of the $\overline{\mathcal{T}}(\lambda )$-spectrum $\Sigma^\mathsf{(6VD)}$ that $\bar{\mathsf{t}}(\lambda )\in\Sigma^\mathsf{(6VD)}$ if and only if $-\bar{\mathsf{t}}(\lambda )\in\Sigma^\mathsf{(6VD)}$.
We moreover recall the formula (5.10) of \cite{LevNT15}, which in the present antiperiodic case reads
\begin{equation}\label{prod-transfer}
\prod_{a=1}^{\mathsf{N}}\overline{\mathcal{T}}(\xi_{a})
=\prod_{a=1}^{\mathsf{N}}\mathsc{a}(\xi_a)\ \prod_{a=1}^{\mathsf{N}}\!\left\{ \mathsf{T}_\tau^{\sigma_a^z}\, \sigma_a^x\right\},
\end{equation}
from which it follows that any $\overline{\mathcal{T}}(\lambda )$-eigenvalue $\bar{\mathsf{t}}(\lambda )$ should satisfy the identity
\begin{equation}
    \prod_{a=1}^\mathsf{N} \bar{\mathsf{t}}(\xi_a )=\pm\prod_{a=1}^{\mathsf{N}}\mathsc{a}(\xi_a).
\end{equation}
The partitioning of the spectrum in terms of the two subsets  $\Sigma^\mathsf{(6VD)}_+$ and $\Sigma^\mathsf{(6VD)}_-$ as defined above hence follows.
The relation between the corresponding $+$ and $-$ eigenstates is then a consequence of the symmetry property \eqref{z-Tbar}.
\end{proof}

\begin{lemma}
The spectrum $\Sigma_{\mathsf{(0,0)}}^\mathsf{(8V)}$ of the periodic 8-vertex transfer matrix $\mathsf{T}^{\mathsf{(8V)}}_\mathsf{(0,0)}(\lambda )$ in $\mathbb{V}_\mathsf{N}$ for $\mathsf{N}$ odd is such that
\begin{equation}\label{intersect-0}
    \Sigma_{\mathsf{(0,0)}}^\mathsf{(8V)}\cap\Sigma^\mathsf{(6VD)}_-=\emptyset ,
\end{equation}
where 
$\Sigma^\mathsf{(6VD)}_-$ is defined as in Lemma~\ref{lem-sp-6VD}.
\end{lemma}

\begin{proof}
This follows from the fact that any $\mathsf{t}(\lambda)\in\Sigma_{\mathsf{(0,0)}}^\mathsf{(8V)}$ satisfies the relation
\begin{equation}
\prod_{a=1}^{\mathsf{N}}\mathsf{t}(\xi_a)
=\prod_{a=1}^{\mathsf{N}}\mathsc{a}(\xi_a),
\label{8V-global}
\end{equation}
which is the analog for the periodic 8-vertex transfer matrix of formula (5.10) of \cite{LevNT15}.
\end{proof}

\begin{lemma}\label{lem-ker}
For $\mathsf{N}$ odd, the kernels $\ker\mathbf{S}^{(0)}$ and $\ker\hat{\mathbf{S}}^{(0)}$ of the operators $\mathbf{S}^{(0)}$ and $\hat{\mathbf{S}}^{(0)}$ are the $2^{\mathsf{N}-1}$ dimensional linear subspaces of $\mathbb{V}_\mathsf{N}$ respectively characterized by
\begin{align}
       &\ker\mathbf{S}^{(0)}
       =\left\{ \mathbf{P}^{(0)}\, \ket{  \Psi_{\bar{\mathsf{t}}_-}^\mathsf{(6VD)}  }\  \Big|\
       \ket{  \Psi_{\bar{\mathsf{t}}_-}^\mathsf{(6VD)}  } \in\mathbf{\Sigma}_-^\mathsf{(6VD)}
       \right\} \label{kerS}
       \\
       &\ker\hat{\mathbf{S}}^{(0)}
       =\left\{ \mathbf{P}^{(0)}\, \ket{ \Psi_{\bar{\mathsf{t}}_+}^\mathsf{(6VD)} }\  \Big|\
       \ket{ \Psi_{\bar{\mathsf{t}}_+}^\mathsf{(6VD)} }\in\mathbf{\Sigma}_+^\mathsf{(6VD)}
       \right\} \label{kerSbar}
\end{align}
i.e., $\ker(\mathbf{S}^{(0)}\mathbf{P}^{(0)})=\mathbf{\Sigma}_-^\mathsf{(6VD)}$ and $\ker(\hat{\mathbf{S}}^{(0)}\mathbf{P}^{(0)})=\mathbf{\Sigma}_+^\mathsf{(6VD)}$.
\end{lemma}

\begin{proof}
The identity \eqref{intersect-0}, together with the relations \eqref{right-act-T2} and \eqref{right-act-T3}, imply that, for any $\ket{  \Psi_{\bar{\mathsf{t}}_-}^\mathsf{(6VD)}  } \in\mathbf{\Sigma}_-^\mathsf{(6VD)}$ and any $\ket{ \Psi_{\bar{\mathsf{t}}_+}^\mathsf{(6VD)} }\in\mathbf{\Sigma}_+^\mathsf{(6VD)}$,
\begin{equation}
    \mathbf{S}^{(0)}\,\mathbf{P}^{(0)}\, \ket{\Psi_{\bar{\mathsf{t}}_-}^\mathsf{(6VD)} }=0,
    \qquad
    \hat{\mathbf{S}}^{(0)}\,\mathbf{P}^{(0)}\,\ket{\Psi_{\bar{\mathsf{t}}_+}^\mathsf{(6VD)} }=0,
\end{equation}
so that $\mathbf{\Sigma}_-^\mathsf{(6VD)}\subset \ker(\mathbf{S}^{(0)}\mathbf{P}^{(0)})$ and $\mathbf{\Sigma}_+^\mathsf{(6VD)}\subset \ker(\hat{\mathbf{S}}^{(0)}\mathbf{P}^{(0)})$.
It therefore remains to prove that $\ker\mathbf{S}^{(0)}\cap \ker\hat{\mathbf{S}}^{(0)}=\{0\}$.

To this aim, let us consider the explicit form of the operators $\mathbf{S}^{(0)}$ and $\hat{\mathbf{S}}^{(0)}$ when acting on the local spin basis of $\mathbb{V}_\mathsf{N}$. It is easy to see, from \eqref{S-spin-R} and the definition \eqref{def-Sbar} of $\hat{\mathbf{S}}^{(0)}$, that
\begin{align}
   &\mathbf{S}^{(0)}\Big(\underset{n=1}{\overset{\mathsf{N}}{\otimes}}\!\ket{n,h_{n}}\Big)
      =    \underset{n=1}{\overset{\mathsf{N}}{\otimes}}\left[ S_n\Big(\xi_n+\frac{\eta}{2}\, \Big|\, \hat{t}_{0,\mathbf{h}}^{(n)}\Big)\, \ket{n,h_{n}} \right],\label{act-S}\\
  &\hat{\mathbf{S}}^{(0)}\Big(\underset{n=1}{\overset{\mathsf{N}}{\otimes}}\!\ket{n,h_{n}}\Big)
      =    \underset{n=1}{\overset{\mathsf{N}}{\otimes}}\left[ \hat{S}_n\Big(\xi_n+\frac{\eta}{2}\, \Big|\, \hat{t}_{0,\mathbf{h}}^{(n)}\Big)\, \ket{n,h_{n}} \right],\label{act-Sbar}
\end{align}
where $ \hat{t}_{0,\mathbf{h}}^{(n)}$ is given by \eqref{thatj} in the case $\mathsf{x}=\mathsf{y}=r=0$, where $S(\lambda|t)$ is the $2\times 2$ matrix \eqref{mat-S} (with here $\mathsf{y}=0$), and where
\begin{equation}
    \hat{S}(\lambda|t) =  S(\lambda|t)\, \sigma^z
    = \begin{pmatrix}
     \theta_2(-\lambda+t|2\omega) &&-\theta_2(\lambda+t|2\omega) \\
    \theta_3(-\lambda+t|2\omega) && -\theta_3(\lambda+t|2\omega)
    \end{pmatrix}    .             
\end{equation}
Let us define, for each $n\in\{1,\ldots,\mathsf{N}\}$, the following linear subspaces of $\mathbb{V}_\mathsf{N}$:
\begin{equation}
     \mathbb{K}_n^\pm=\mathbb{V}_{\hat{n}}[0]\otimes V_n^\pm,
\end{equation}
where $V_n^\pm$ is the one-dimensional linear subspace of $V_n$ spanned by the vector $(\ket{n,0}\pm\ket{n,1})$, whereas $\mathbb{V}_{\hat{n}}[0]$ is 
the linear subspace of $\mathbb{V}_{\hat{n}}\equiv \otimes_{ j\not= n} V_j$ which cancels the action of the operator $ \mathsf{S}^{+-}_{\hat{n}}\equiv\sum_{k=1}^{n-1}\sigma_k^z-\sum_{k=n+1}^\mathsf{N}\sigma_k^z$ :
\begin{equation}
   \mathbb{V}_{\hat{n}}[0]=\left\{ \ket{\mathbf{v}}\in \mathbb{V}_{\hat{n}} \ \big|\   \mathsf{S}^{+-}_{\hat{n}} \ket{\mathbf{v}}=0 \right\}.
\end{equation}
It is clear from the quasi-tensor form of the action \eqref{act-S}-\eqref{act-Sbar} that
\begin{equation}
   \ker \mathbf{S}^{(0)}=\mathbb{K}_1^-+\mathbb{K}_2^-+\cdots+\mathbb{K}_\mathsf{N}^-,
   \qquad
    \ker \hat{\mathbf{S}}^{(0)}=\mathbb{K}_1^++\mathbb{K}_2^++\cdots+\mathbb{K}_\mathsf{N}^+.
\end{equation}
This can for instance easily be shown by induction, considering partial operators such as \eqref{Sj}.
Note moreover that these sums are in fact direct sums, since any state of the form $(\ket{n,0}\pm\ket{n,1})\otimes\ket{\mathbf{v}}$, for $\ket{\mathbf{v}}\in\mathbb{V}_{\hat{n},\hat{m}}\equiv\otimes_{j\not= m,n}V_j$, $m\not= n$, is {\em not} an eigenstate of  $ \mathsf{S}^{+-}_{\hat{m}}$.
For the same reason, one has $\mathbb{K}_n^+\cap\mathbb{K}_m^-=\{0\}$ if $n\not= m$. Finally, it is clear that $\mathbb{K}_n^+\cap\mathbb{K}_n^-=\{0\}$, so that $\ker\mathbf{S}^{(0)}\cap \ker\hat{\mathbf{S}}^{(0)}=\{0\}$, which proves the equality in \eqref{kerS} and \eqref{kerSbar}.
\end{proof}

Hence we arrive at the following result:

\begin{theorem}\label{th-sp-8V-per}
The spectrum of the periodic 8-vertex transfer matrix $\mathsf{T}_{\mathsf{(0,0)}}^{\mathsf{(8V)}}(\lambda )$ for $\mathsf{N}$ odd is 
\begin{equation}
    \Sigma_{\mathsf{(0,0)}}^\mathsf{(8V)} = \Sigma^\mathsf{(6VD)}_+,
\end{equation}
where $\Sigma^\mathsf{(6VD)}_+$ is the `+' part \eqref{Sigma+} of the antiperiodic dynamical 6-vertex transfer matrix spectrum $\Sigma^\mathsf{(6VD)}$. Each of the $2^{\mathsf{N}-1}$  $\mathsf{T}_{\mathsf{(0,0)}}^{\mathsf{(8V)}}(\lambda )$-eigenvalues $\mathsf{t}(\lambda)\in\Sigma_{\mathsf{(0,0)}}^\mathsf{(8V)} = \Sigma^\mathsf{(6VD)}_+$ is doubly degenerated, with two linearly independent $\mathsf{T}_{\mathsf{(0,0)}}^{\mathsf{(8V)}}(\lambda )$-eigenvectors given by
\begin{align}
   &\ket{ \Psi_\mathsf{t}^+} = \mathbf{S}^{(0)}\,\mathbf{P}^{(0)}\,\ket{ \Psi_{\mathsf{t}}^\mathsf{(6VD)} },
   \label{psi+}\\
   &\ket{ \Psi_\mathsf{t}^-} = \begin{cases}
      \Gamma_x\,\ket{ \Psi_\mathsf{t}^+} \quad &\text{if $\ket{ \Psi_\mathsf{t}^+}$ is a $\Gamma_z$-eigenstate},\\
      \Gamma_z\,\ket{ \Psi_\mathsf{t}^+} \quad &\text{otherwise},
      \end{cases}
      \label{psi-}
\end{align}
where $\ket{ \Psi_{\mathsf{t}}^\mathsf{(6VD)} }$ denotes the $\overline{\mathcal{T}}(\lambda )$-eigenvector  with eigenvalue $\mathsf{t}(\lambda )$, and where $\Gamma_x=\otimes_{n=1}^\mathsf{N}\sigma_n^x$.
\end{theorem}

\begin{proof}
From \eqref{right-act-T2} and \eqref{right-act-T3} one obtains that, for each $\overline{\mathcal{T}}(\lambda )$-eigenvector $\ket{ \Psi_{\mathsf{t}}^\mathsf{(6VD)} }\in\mathbf{\Sigma}_+^\mathsf{(6VD)} $ with $\overline{\mathcal{T}}(\lambda )$-eigenvalue $\mathsf{t}(\lambda)\in \Sigma^\mathsf{(6VD)}_+$,
\begin{equation*}
    \mathsf{T}_{\mathsf{(0,0)}}^{\mathsf{(8V)}}(\lambda )\, \mathbf{S}^{(0)}\mathbf{P}^{(0)}\ket{\Psi_{\mathsf{t}}^\mathsf{(6VD)}}
    = \mathbf{S}^{(0)}\mathbf{P}^{(0)}\, \overline{\mathcal{T}}(\lambda)\, \ket{\Psi_{\mathsf{t}}^\mathsf{(6VD)}}
    =\mathsf{t}(\lambda)\ \mathbf{S}^{(0)}\mathbf{P}^{(0)}\ket{\Psi_{\mathsf{t}}^\mathsf{(6VD)}}.
\end{equation*}
%
From Lemma~\ref{lem-ker}, $\mathbf{P}^{(0)}\ket{\Psi_{\mathsf{t}}^\mathsf{(6VD)}}\notin\ker\mathbf{S}^{(0)}$ so that $\ket{ \Psi_\mathsf{t}^+}$ \eqref{psi+} is a nonzero eigenvector of $\mathsf{T}_{\mathsf{(0,0)}}^{\mathsf{(8V)}}(\lambda )$ associated with the eigenvalue $\mathsf{t}(\lambda)$.
As both $\Gamma_z$ and $\Gamma_x$ commute with $\mathsf{T}_{\mathsf{(0,0)}}^{\mathsf{(8V)}}(\lambda )$ then $\ket{ \Psi_\mathsf{t}^-}$ \eqref{psi-} is clearly also a nonzero eigenvector of $\mathsf{T}_{\mathsf{(0,0)}}^{\mathsf{(8V)}}(\lambda )$ with eigenvalue $\mathsf{t}(\lambda)$.
Let us observe now that the statement $\ket{ \Psi_\mathsf{t}^+}$ is not a $\Gamma_z$-eigenvector is just equivalent to say that $\ket{ \Psi_\mathsf{t}^+}$ and $\Gamma_z\,\ket{ \Psi_\mathsf{t}^+}$ are two linearly independent states. Instead if $\ket{ \Psi_\mathsf{t}^+}$ is a $\Gamma_z$-eigenvector then also $\Gamma_x\,\ket{ \Psi_\mathsf{t}^+}$ is a $\Gamma_z$-eigenvector but with opposite eigenvalue.
These observations just show that in all possible cases  $\ket{ \Psi_\mathsf{t}^\pm}$ \eqref{psi+}-\eqref{psi-} are two independent $\mathsf{T}_{\mathsf{(0,0)}}^{\mathsf{(8V)}}(\lambda )$-eigenvectors associated with the same eigenvalue $\mathsf{t}(\lambda)$. Hence we have constructed in that way a family of $2\times 2^{\mathsf{N}-1}=2^\mathsf{N}$ linearly independent eigenvectors of $\mathsf{T}_{\mathsf{(0,0)}}^{\mathsf{(8V)}}(\lambda )$, so that we have a complete description.
\end{proof}

Theorem \ref{th-sp-8V-per} allows us in particular to define an
invertible linear operator $\mathbf{G}$ on $\mathbb{V}_{\mathsf{N}}$ by the
following action on the eigenbasis \eqref{psi+}-\eqref{psi-} of $\mathsf{T}_{\left( 0,0\right) }^{\mathsf{(8V)}}(\lambda )$:
\begin{equation}
\mathbf{G}\,\ket{ \Psi_\mathsf{t}^\pm}\equiv \ket{ \Psi_\mathsf{t}^\mp},
\qquad \forall \,\mathsf{t}\in\Sigma_{\mathsf{(0,0)}}^\mathsf{(8V)},
\end{equation}
so that we can define a modified version $\mathbf{\bar S}^{(0)}=\mathbf{G}\,\mathbf{S}^{(0)}\,\Gamma _{z}$ of the vertex-IRF transformation. This enables us to formulate the following lemma:

\begin{lemma}
Let $\mathsf{(x,y)}=(0,0)$ and $\mathsf{N}$ be odd.
Then the periodic 8-vertex transfer matrix is related to the analog \eqref{dyn-transfer} on $\mathbb{V}_{\mathsf{N}}$ of the antiperiodic dynamical 6-vertex transfer matrix by
\begin{equation}
\mathsf{T}^{\mathsf{(8V)}}_\mathsf{(0,0)}(\lambda )\,\mathbf{S}^{(\pm)}
=\mathbf{S}^{(\mp) }\,\overline{\mathsf{T}}^{\mathsf{(6VD)}}(\lambda )
\label{Invertible-pseudo-similarities}
\end{equation}
where $\mathbf{S}^{(\pm)}$ are invertible endomorphisms on $\mathbb{V}_{\mathsf{N}}$ defined by
\begin{equation}
\mathbf{S}^{(\pm) }=\mathbf{S}^{(0)}\pm \mathbf{\bar S}^{(0)}.
\end{equation}
%
\end{lemma}

\begin{proof}
It follows from the previous study that the periodic 8-vertex transfer matrix is related to the transfer matrix \eqref{dyn-transfer} by the two following identities:
\begin{equation}
\mathsf{T}_{( 0,0) }^{\mathsf{(8V)}}(\lambda )\,\mathbf{S}^{(0)}
=\mathbf{S}^{(0)}\,\overline{\mathsf{T}}^{\mathsf{(6VD)}}(\lambda ),
\qquad
\mathsf{T}_{( 0,0) }^{\mathsf{(8V)}}(\lambda )\, \mathbf{\bar S}^{(0)}
=-\mathbf{\bar S}^{(0)}\,\overline{\mathsf{T}}^{\mathsf{(6VD)}}(\lambda ).
\end{equation}
Taking their sum and difference we therefore get \eqref{Invertible-pseudo-similarities}.
It remains to prove that $\mathbf{S}^{(\pm)}$ are invertible. In order to do so we
observe that, for any $\bar{\mathsf{t}}_\epsilon(\lambda )\in\Sigma^\mathsf{(6VD)}_\epsilon$, $\epsilon\in\{+,-\}$, we have by definition
\begin{align}
  \mathbf{S}^{(\pm)}\,\mathbf{P}^{(0)}\,\ket{\Psi_{\bar{\mathsf{t}}_\epsilon}^\mathsf{(6VD)} } 
  &=\frac{1+\epsilon }{2}\,\mathbf{S}^{(0)}\,\mathbf{P}^{(0)}\,\ket{\Psi_{\bar{\mathsf{t}}_+}^\mathsf{(6VD)} }
       \pm \frac{1-\epsilon }{2}\,\mathbf{\bar S}^{(0)}\,\mathbf{P}^{(0)}\,\ket{\Psi_{\bar{\mathsf{t}}_-}^\mathsf{(6VD)} }
        \\
 &=\frac{1+\epsilon }{2}\,\ket{ \Psi_\mathsf{t}^+}\pm \frac{1-\epsilon }{2}\,\ket{ \Psi_\mathsf{t}^-}
 =(\pm 1)^\frac{1-\epsilon}{2}\, \ket{\Psi_\mathsf{t}^\epsilon},
\end{align}
where $\ket{ \Psi_\mathsf{t}^\pm}$ are the $\mathsf{T}_{( 0,0) }^{\mathsf{(8V)}}(\lambda)$-eigenstates \eqref{psi+}-\eqref{psi-} associated with the eigenvalue $\mathsf{t}(\lambda)=\epsilon\bar{\mathsf{t}}_\epsilon(\lambda )$.
So both $\mathbf{S}^{(+)}$ and $\mathbf{S}^{(-)}$ are invertible as they transform the eigenbasis of $\overline{\mathsf{T}}^{\mathsf{(6VD)}}(\lambda )$ into that of $\mathsf{T}_{\left( 0,0\right) }^{\mathsf{(8V)}}(\la)$.
\end{proof}

This result enables us to use the SOV construction \cite{Nic13a,LevNT15} for the antiperiodic dynamical 6-vertex model to explicitly construct the periodic 8-vertex transfer matrix eigenstates and to characterize its spectrum. As in Section~\ref{ssec-spectr-8V}, the eigenstates can be defined in a self-contained way in terms of the basis \eqref{states-bar} of $\mathbb{V}_\mathsf{N}$. 

\begin{theorem}
\label{th-T8V-per}
Let $\mathsf{N}$ be odd.
For any fixed $\mathsf{N}$-tuple of
inhomogeneities $(\xi _{1},\ldots,\xi _{\mathsf{N}})\in \mathbb{C}^{\mathsf{N}}$ satisfying \eqref{cond-inh0}-\eqref{cond-inh}, the spectrum $\Sigma_{\mathsf{(0,0)}}^\mathsf{(8V)}$ of the periodic 8-vertex transfer matrix $\mathsf{T}^{\mathsf{(8V)}}_\mathsf{(0,0)}(\lambda )$ \eqref{transfer} is doubly degenerated and coincides with the set of functions of the form
\begin{equation}
   \mathsf{t}(\lambda )\equiv 
   \sum_{a=1}^{\mathsf{N}} 
\frac{\theta(t_{0,\mathbf{0}}-\lambda +\xi _{a})}{\theta (t_{0,\mathbf{0}})}
\prod_{b\neq a}\frac{\theta (\lambda -\xi _{b})}{\theta (\xi _{a}-\xi _{b})}\,
 \mathsf{t}(\xi_a) ,
\label{t-8V00}
\end{equation}
which satisfy the discrete system of equations
\begin{equation}
\mathsf{t}(\xi _{a})\,\mathsf{t}(\xi_{a}-\eta)=
\mathsc{a}(\xi _{a})\,\mathsc{d}(\xi _{a}-\eta),
\qquad
\forall a\in \{1,\ldots,\mathsf{N}\},
\label{eq-8V}
\end{equation}
and the condition
\begin{equation}\label{cond-+}
    \prod_{n=1}^\mathsf{N}\mathsf{t}(\xi_n )=\prod_{n=1}^\mathsf{N}\mathsc{a}(\xi_n).
\end{equation}
A basis of the $\mathsf{T}^{\mathsf{(8V)}}_\mathsf{(0,0)}(\lambda )$-eigenspace associated with the eigenvalue $\mathsf{t}(\lambda ) \in \Sigma_{\mathsf{(0,0)}}^\mathsf{(8V)}$ is provided by the two vectors  $\ket{ \Psi_\mathsf{t}^+ }$ and $\ket{ \Psi_\mathsf{t}^- }$ in $\mathbb{V}_\mathsf{N}$ defined by
\begin{align}
& \ket{ \Psi_\mathsf{t}^\epsilon }
=(\pm1)^\frac{1-\epsilon }{2}\sum_{\mathbf{h}\in\{0,1\}^\mathsf{N}}
\prod_{a=1}^{\mathsf{N}}\bigg[ \bigg(\epsilon\, \frac{\mathsc{a}(\xi_a)}{\mathsc{d}(\xi_a-\eta)}\bigg)^{\! h_a}\,  \mathsf{q}^{(h_a)}_{\mathsf{t},a}\bigg]\,
\det_{\mathsf{N}}\big[\Theta^{(\mathbf{h}) }\big]\ 
\mathbf{S}^{(\pm)}\ket{\underline{\mathbf{h}}} , 
 \label{eigen+XYZ}\\
&  \bra{ \Psi_\mathsf{t}^\epsilon}
=(\pm1)^\frac{1-\epsilon }{2}\sum_{\mathbf{h}\in\{0,1\}^\mathsf{N}}
\prod_{a=1}^{\mathsf{N}}\big( \epsilon^{ h_a}\,  \mathsf{q}^{(h_a)}_{\mathsf{t},a}\big)\,
\det_{\mathsf{N}}\big[\Theta^{(\mathbf{h}) }\big]\ 
\,\bra{\underline{\mathbf{h}}}\big[\mathbf{S}^{(\pm)}\big]^{-1} , 
 \label{eigen-XYZ}
\end{align}
where $\epsilon=+,-$ and the coefficients $\mathsf{q}_{\mathsf{t},a}^{(h_a)}$ are (up to an overall normalization) characterized by
\begin{equation}
\frac{ \mathsf{q}_{\mathsf{t},a}^{(1)} }{  \mathsf{q}_{\mathsf{t},a}^{(0)}  }
=\frac{\mathsc{d}(\xi_a-\eta)}{ \mathsf{t}(\xi_a-\eta)}
= \frac{ \mathsf{t}(\xi_a)}{\mathsc{a}(\xi_a)} .
\label{t-q-XYZbis}
\end{equation}
\end{theorem}

Once again, one can use the results of \cite{LevNT15} to reformulate the characterization of Theorem~\ref{th-T8V-per} in terms of solutions of functional $T$-$Q$ equations i.e., in terms of solutions of Bethe-type equations. It is for example possible to use the reformulation in terms of the solutions of the functional $T$-$Q$ equations with an extra inhomogeneous terms which was shown to provide a complete description of the dynamical 6-vertex spectrum and eigenstates, as described in Appendix~B of \cite{LevNT15}. This reformulation allows one to rewrite the transfer matrix eigenstates in a form similar to ABA (see Appendix~\ref{app-inhom}), and to deal with the homogeneous limit. However, as already mentioned in Section~\ref{sec-eq-Baxter}, it is for the moment not so obvious whether one can efficiently analyze the thermodynamic limit of the corresponding Bethe-type equations.
Alternatively we can use, thanks to the gauge transformations, the ansatz proposed in \cite{LevNT15} 
for the solutions of the homogeneous $T$-$Q$ functional equation. This ansatz consists in looking for the solutions of Baxter's homogeneous $T$-$Q$ functional equation within particular classes of elliptic polynomials with different quasi-periods compatible with those of the transfer matrix and of the other coefficients of the $T$-$Q$ equation (see \cite{LevNT15} for details), i.e. in postulating some particular form of the Bethe equations and of the SOV eigenstates.
In other words, it enables us {\em a priori} to write some kind of ABA-type expressions for the transfer matrix spectrum and eigenstates, similarly as what as been done in Proposition~\ref{Cor-Eigen-Bethe} in the twisted case.
However, until now, the completeness of this ``algebraic Bethe-type'' ansatz has been proven for even $\mathsf{N}$ only, and therefore it is still an open question whether it provides a good characterization of the periodic 8-vertex transfer matrix spectrum and eigenstates in the present odd $\mathsf{N}$ case. We nevertheless expect that the restriction to the case of even $\mathsf{N}$ is purely technical (the proof of the completeness happens to be simpler in that case, see \cite{LevNT15}), and we forecast to come back to this completeness problem for odd $\mathsf{N}$ in a further publication. 

\section{Conclusion}

In this paper we have studied the finite-size 8-vertex (or XYZ) model for different types of integrable quasi-periodic boundary conditions which are not solved by Bethe ansatz: with a twist by $\sigma^\alpha$, $\alpha=x,y,z$, or without twist (periodic case) but, in the latter case, for a model with an odd number of sites only. In all these cases, we have shown the relation, by means of the vertex-IRF transformation, with the dynamical 6-vertex model (also called SOS model) with certain types of {\em antiperiodic} boundary conditions, a model solvable by Sklyanin's Separation of Variables approach.

We have shown that the vertex-IRF transformation relating these two models is bijective in all the twisted cases, so that we can directly use the known SOV description \cite{LevNT15} of the antiperiodic dynamical 6-vertex transfer matrix spectrum and eigenstates to characterized the spectrum and eigenstates of the 8-vertex twisted transfer matrix. We have therefore obtained a complete description of this spectrum, which in that case is simple, and of the corresponding eigenstates, which can be constructed by means of a SOV basis, in terms of solutions of a discrete system of equations involving the inhomogeneity parameters of the model. Still using the known results for the antiperiodic dynamical 6-vertex model \cite{LevNT15}, we have reformulated this characterization in terms of solutions of some homogeneous functional $T$-$Q$ equation of Baxter's type, i.e. in terms of solutions of a system of Bethe-type equations, at least in the case of $\mathsf{N}$ even.

In the periodic case with an odd number of sites, the vertex-IRF transformation is not bijective and the  spectrum of the 8-vertex transfer matrix is not simple. This case was partially studied in \cite{Nic13a}. In the present paper, we have considered two variants of the vertex-IRF transformation with non-intersecting kernels from which we were able to completely describe the 8-vertex transfer matrix spectrum and eigenstates. We have shown that the spectrum is doubly degenerated and coincides with half of the spectrum of the corresponding antiperiodic dynamical 6-vertex transfer matrix. To each of the eigenvalues correspond two linearly independent eigenstates which can be obtained  from the antiperiodic dynamical 6-vertex transfer matrix ones by means of one or the other variants of the vertex-IRF transformation. As in the twisted case, these eigenstates can be constructed by means of the dynamical 6-vertex SOV basis, in terms of solutions of a system of discrete equations involving the inhomogeneity parameters of the model. It would be interesting to be able to also prove in that case (odd number of sites) the equivalence of this characterization with a formulation  in terms of solutions of Bethe-type equations: an ansatz has been made in \cite{LevNT15}, but the completeness of this ansatz has for the moment not been proven in the case of a system with an odd number of sites.

To conclude, let us briefly mention the question of the computation of the finite-size form factors and of the correlation functions of the model, which is a difficult problem due to the combinatorial complexity of the vertex-IRF transformation. A natural strategy to compute the finite-size form factors would be to use the solution of the quantum inverse problem \eqref{inv-pb1}-\eqref{inv-pb2} (which can easily be adapted to any twisted case following the lines presented in \cite{LevNT15}) so as to express the action of the local operators on the twisted 8-vertex eigenstates in terms of the action of the corresponding elements of the twisted 8-vertex monodromy matrix or of its inverse. The latter can then be related to the action of the elements of the dynamical 6-vertex monodromy matrix using the relations \eqref{act-el1}-\eqref{act-el2}. However, these relations involve a dressing of the off-diagonal dynamical 6-vertex matrix elements by a non-trivial action of the vertex-IRF transformation on the whole space of states, which makes it difficult the obtention of generic and compact formulas for the corresponding form factors. We plan to consider this challenging problem in a future work.

\section*{Acknowledgements}

G.N. and V.T. are supported by CNRS.
We also acknowledge  the support from the ANR grant DIADEMS 10 BLAN 012004 for early stage of this work.
G.N. would like to express his gratefulness to the organizers of the conference ``Baxter 2015: Exactly Solved Models $\&$ Beyond
ANU, Palm Cove'' for the invitation and for the great honor to present this paper there.


\appendix

\section{Diagonalization of the antiperiodic dynamical 6-vertex transfer matrix by SOV}
\label{app-spectr-6VD}

The $\mathsf{(x,y)}$-dynamical 6-vertex model with antiperiodic boundary condition has been studied in \cite{LevNT15} by means of Sklyanin's quantum separation of variable approach \cite{Skl90,Skl92}.
For completeness,  we briefly recall in this appendix some of the results of \cite{LevNT15} that we use in the present article concerning the SOV characterization of the corresponding transfer matrix eigenvalues and eigenstates.

In the subspace $\mathbb{\bar{D}}_{\mathsf{(6VD)},\mathsf{N}}^{(0,\mathcal{L}/\mathcal{R})}$ of $\mathbb{D}_{\mathsf{(6VD)},\mathsf{N}}^{\mathcal{L}/\mathcal{R}}$, we define the following left and right {\em reference states}:%
\begin{equation}\label{ref-states}
  \bra{\mathbf{0} }
\equiv \frac{1}{\mathsc{n}}
\big(\otimes _{n=1}^{\mathsf{N}}\langle n,h_{n}=0| \big)\otimes \langle t_{0,\mathbf{0}}|,
\qquad
  \ket{\mathbf{1} }
 \equiv \frac{1}{\mathsc{n}}
 \big( \otimes _{n=1}^{\mathsf{N}}|n,h_{n}=1\rangle \big)
\otimes |t_{0,\mathbf{1}}\rangle ,
\end{equation}
where we have used the notations 
$\mathbf{0}\equiv (h_{1}=0,\ldots,h_{\mathsf{N}}=0)$ and 
$\mathbf{1}\equiv (h_{1}=1,\ldots,h_{\mathsf{N}}=1)$, and where $\mathsc{n}$ is a conveniently chosen normalization constant.
Then, for each $\mathsf{N}$-tuple $\mathbf{h}\equiv (h_{1},\ldots,h_{\mathsf{N}})\in \{0,1\}^{\mathsf{N}}$, we construct a state $\bra{\mathbf{h}}\in\bar{\mathbb{D}}_{\mathsf{(6VD)},\mathsf{N}}^{(0,\mathcal{L})}$ and a state $\ket{ \mathbf{h} } \in\bar{\mathbb{D}}_{\mathsf{(6VD)},\mathsf{N}}^{(0,\mathcal{R})}$ as\footnote{We slightly simplify here the notations with respect to \cite{LevNT15}: the states \eqref{D-left-eigenstates} and \eqref{D-right-eigenstates} are respectively denoted $\bra{0,\mathbf{h}}$ and $\ket{\mathbf{h},0 } $ in \cite{LevNT15}.}
\begin{align}
 &\bra{\mathbf{h}}\equiv 
    \bra{\mathbf{0}} \prod_{n=1}^{\mathsf{N}}
 \left( \frac{\mathcal{C}(\xi_{n} )}{\text{\textsc{d}}(\xi _{n}-\eta )}\right)^{h_{n}},
\label{D-left-eigenstates}\\
 &\ket{\mathbf{h} } 
 \equiv \prod_{n=1}^{\mathsf{N}}
 \left( \frac{\mathcal{C}(\xi _{n}-\eta  )}{\text{\textsc{d}}(\xi _{n}-\eta )}\right) ^{(1-h_{n})}
 \ket{\mathbf{1} } .
 \label{D-right-eigenstates}
\end{align}
Under the hypothesis \eqref{cond-inh} on the inhomogeneity parameters $\xi _{1},\ldots,\xi _{\mathsf{N}}\in \mathbb{C}$ of the model,
it has been proven in \cite{LevNT15} that the set of vectors \eqref{D-left-eigenstates} (respectively \eqref{D-right-eigenstates}) defines a SOV basis of $\mathbb{\bar{D}}_{\mathsf{(6VD)},\mathsf{N}}^{(0,\mathcal{L}/\mathcal{R})}$ for the operator entries of the antiperiodic dynamical 6-vertex monodromy matrix (see Theorem~3.1 and Theorem~3.2 of \cite{LevNT15}).
These two basis are orthogonal, and the normalization constant $\mathsc{n}$ can be chosen in such a way that, $\forall \mathbf{h},\mathbf{k}\in\{0,1\}^\mathsf{N}$,
\begin{equation}\label{sc-rh}
   \moy{ \mathbf{h}\, |\, \mathbf{k} }
   = \delta_{\mathbf{h},\mathbf{k}}\,
   \frac{e^{-i\mathsf{y}\eta \sum_{j=1}^\mathsf{N}h_j}}
           {\det_{\mathsf{N}}\big[ \Theta^{(\mathbf{h})} \big]},
\end{equation}
in terms of the $\mathsf{N}\times\mathsf{N}$ matrices $\Theta^{(\mathbf{h})}$ of elements
\begin{equation}\label{mat-Theta}
  \big[ \Theta^{(\mathbf{h})} \big]_{ij}=\vartheta_{j-1}(\xi_i-h_i\eta-\bar\xi_0 ),
   \qquad
   \text{with}
   \quad
   \bar\xi_0=\frac{1}{\mathsf{N}}\left(\sum_{k=1}^\mathsf{N}\xi_k+t_{0,\mathbf{0}}\right),
\end{equation}
and where the functions $\vartheta_j$ are defined as
\begin{equation}\label{def-theta_j}
  \vartheta_j(\lambda)=\sum_{n\in\mathbb{Z}} 
  e^{i\pi\mathsf{N}\omega(n+\frac12-\frac{j}{\mathsf{N}})^2+2i\mathsf{N}(n+\frac12-\frac{j}{\mathsf{N}})(\lambda-\frac{\pi}{2})},
  \quad 0\le j\le \mathsf{N}-1.
\end{equation}

Moreover, on this basis, the action of the operators $\mathcal{B}(\lambda)$ and $\mathcal{C}(\lambda)$ have been explicitly  computed in a quasi-local way:
\begin{align}
&\bra{ \mathbf{h}}\,\mathcal{C}(\lambda  ) 
=\sum_{a=1}^{\mathsf{N}}
e^{i\mathsf{y}(\xi _{a}^{(h_{a})}-\lambda )}\,
\frac{\theta (t_{0,\mathbf{h}}-\lambda +\xi _{a}^{(h_{a})})}{\theta (t_{0,\mathbf{h}})}
  \prod_{b\neq a}\frac{\theta (\lambda -\xi_{b}^{(h_{b})})}{\theta (\xi _{a}^{(h_{a})}-\xi _{b}^{(h_{b})})}\,
  \text{\textsc{d}}(\xi _{a}^{(1-h_{a})})\,
  \bra{ \mathsf{T}_{a}^{+} \mathbf{h}}\, , 
 \label{C-SOV_D-left}
  \\
&\bra{ \mathbf{h}}\, \mathcal{B}(\lambda  )
 =\sum_{a=1}^{\mathsf{N}}e^{i\mathsf{y}(\xi _{a}^{(h_{a})}-\lambda )}\,
   \frac{\theta (t_{0,\mathbf{h}}-\lambda +\xi_{a}^{(h_{a})})}{\theta (t_{0,\mathbf{h}})} 
   \prod_{b\neq a}\frac{\theta (\lambda -\xi _{b}^{(h_{b})})}{\theta (\xi _{a}^{(h_{a})}-\xi_{b}^{(h_{b})})}\,
   \mathsc{a}_{\mathsf{x},\mathsf{y}}(\xi_{a}^{(1-h_{a})})\,
   \bra{ \mathsf{T}_{a}^{-}\mathbf{h}}\, ,
\label{B-SOV_D-left}
\end{align}
and
\begin{align}
&\mathcal{C}(\lambda  )\, \ket{\mathbf{h} }
 =\sum_{a=1}^{\mathsf{N}} e^{i\mathsf{y}(\xi _{a}^{(h_{a})}-\lambda )}\,
\frac{\theta (t_{0,\mathbf{h}}-\lambda +\xi _{a}^{(h_{a})})}{\theta(t_{0,\mathbf{h}})}
\prod_{b\neq a}\frac{\theta (\lambda -\xi_{b}^{(h_{b})})}{\theta (\xi _{a}^{(h_{a})}-\xi _{b}^{(h_{b})})}\,\text{\textsc{d}}(\xi _{a}^{(h_{a})})\,
\ket{\mathsf{T}_{a}^{-}\mathbf{h} }\,  ,  
\label{C-SOV_D-right} 
 \\
&\mathcal{B}(\lambda  )\, \ket{\mathbf{h} }  
=\sum_{a=1}^{\mathsf{N}} e^{i\mathsf{y}(\xi _{a}^{(h_{a})}-\lambda )}\,
\frac{\theta (t_{0,\mathbf{h}}-\lambda +\xi _{a}^{(h_{a})})}{\theta(t_{0,\mathbf{h}})}
\prod_{b\neq a}\frac{\theta (\lambda -\xi_{b}^{(h_{b})})}{\theta (\xi _{a}^{(h_{a})}-\xi _{b}^{(h_{b})})}\,
\mathsc{a}_{\mathsf{x},\mathsf{y}}(\xi _{a}^{(h_{a})})\,
 \ket{\mathsf{T}_{a}^{+}\mathbf{h} } \, ,
\label{B-SOV_D-right}
\end{align}
with the notations
\begin{align}
&\xi_a^{(h_a)}=\xi_a-\eta h_a,\\
&\mathsf{T}_{a}^{\pm }(h_{1},\ldots,h_{\mathsf{N}})
=(h_{1},\ldots,h_{a}\pm 1,\ldots,h_{\mathsf{N}}),
\\
&\mathsc{a}_{\mathsf{x},\mathsf{y}}(\lambda )
=(-1)^{\mathsf{x}+\mathsf{y}+\mathsf{x}\mathsf{y}}\, \mathsc{a}(\lambda ).
\end{align}
Finally, we recall Theorem~4.1 of \cite{LevNT15}:

\begin{theorem}[\cite{LevNT15}]\label{thm-eigen-t} 
The antiperiodic dynamical 6-vertex transfer matrix $\overline{\mathcal{T}}(\lambda  )$ \eqref{6VD-antiT} defines a one-parameter family of commuting operators on $\vspace{-1mm}\mathbb{\bar{D}}_{\mathsf{(6VD)},\mathsf{N}}^{(0,\mathcal{L}/\mathcal{R})}$.
For any fixed $\mathsf{N}$-tuple of inhomogeneities $(\xi _{1},\ldots,\xi _{\mathsf{N}})\in 
\mathbb{C}^{\mathsf{N}}$ satisfying \eqref{cond-inh}, the spectrum $\Sigma^\mathsf{(6VD)}$ of $\overline{\mathcal{T}}(\lambda )$ in $\mathbb{\bar{D}}_{\mathsf{(6VD)},\mathsf{N}}^{(0,\mathcal{L}/\mathcal{R})}$ is simple and coincides with the set of functions of the form
\begin{equation}
\bar{\mathsf{t}}(\lambda )
=\sum_{a=1}^{\mathsf{N}} e^{i\mathsf{y}(\xi_a-\lambda)}\,
\frac{\theta(t_{0,\mathbf{0}}-\lambda +\xi _{a})}{\theta (t_{0,\mathbf{0}})}
\prod_{b\neq a}\frac{\theta (\lambda -\xi _{b})}{\theta (\xi _{a}-\xi _{b})}\,
\bar{\mathsf{t}}(\xi _{a}) ,
\quad \big(\bar{\mathsf{t}}(\xi _1),\ldots,\bar{\mathsf{t}}(\xi _{\mathsf{N}})\big)\in\mathbb{C}^\mathsf{N},
\label{set-t}
\end{equation}
which satisfy the discrete system of equations
\begin{equation}
\bar{\mathsf{t}}(\xi _{a})\,\bar{\mathsf{t}}(\xi_{a}-\eta)=( -1) ^{\mathsf{x}+\mathsf{y}+\mathsf{x}\mathsf{y}}\,
\mathsc{a}(\xi _{a})\,\mathsc{d}(\xi _{a}-\eta),
\qquad
\forall a\in \{1,\ldots,\mathsf{N}\}.
\label{dis-sys}
\end{equation}
The right $\overline{\mathcal{T}}(\lambda  )$-eigenstate $\ket{  \Psi_{\bar{\mathsf{t}}}^\mathsf{(6VD)}  }\in\mathbb{\bar{D}}_{\mathsf{(6VD)},\mathsf{N}}^{(0,\mathcal{R})}$ and the left $\overline{\mathcal{T}}(\lambda )$-eigenstate $\bra{  \Psi_{\bar{\mathsf{t}}}^\mathsf{(6VD)}  }\in\mathbb{\bar{D}}_{\mathsf{(6VD)},\mathsf{N}}^{(0,\mathcal{L})}$ associated with the eigenvalue $\bar{\mathsf{t}}(\lambda )\in \Sigma^\mathsf{(6VD)}$ are respectively given by
\begin{align}
& \ket{  \Psi_{\bar{\mathsf{t}}}^\mathsf{(6VD)}  }
=\sum_{\mathbf{h}\in\{0,1\}^\mathsf{N}}
\prod_{a=1}^{\mathsf{N}}\bigg[ e^{i\mathsf{y}\eta h_a}\, \bigg( \frac{\mathsc{a}_\mathsf{x,y}(\xi_a)}{\mathsc{d}(\xi_a-\eta)}\bigg)^{\! h_a}\,  \mathsf{q}_{\bar{\mathsf{t}},a}^{(h_a)}\bigg]\,
\det_{\mathsf{N}}\big[\Theta^{(\mathbf{h}) }\big]\, \ket{\mathbf{h}} , 
 \label{eigenR-6VD}\\
&  \bra{  \Psi_{\bar{\mathsf{t}}}^\mathsf{(6VD)}  }
  =\sum_{\mathbf{h}\in\{0,1\}^\mathsf{N}}
\prod_{a=1}^{\mathsf{N}}\left[ e^{i\mathsf{y}\eta h_a}\,  \mathsf{q}_{\bar{\mathsf{t}},a}^{(h_a)}\right]\,
\det_{\mathsf{N}}\big[\Theta^{(\mathbf{h})}\big] \, 
\bra{ \mathbf{h} }, 
 \label{eigenL-6VD}
\end{align}
where the coefficients $\mathsf{q}_{\bar{\mathsf{t}},a}^{(h_a)}$ are (up to an overall normalization) characterized by
\begin{equation}
\frac{ \mathsf{q}_{\bar{\mathsf{t}},a}^{(1)} }{  \mathsf{q}_{\bar{\mathsf{t}},a}^{(0)}  }
=\frac{\mathsc{d}(\xi_a-\eta)}{\bar{\mathsf{t}}(\xi_a-\eta)}
=(-1)^{\mathsf{x}+\mathsf{y}+\mathsf{x}\mathsf{y}}\, \frac{\bar{\mathsf{t}}(\xi_a)}{\mathsc{a}(\xi_a)} .
\label{t-q-6VD}
\end{equation}
\end{theorem}

\section{Complete characterization of the spectrum and eigenstates through the solutions of an inhomogeneous  $T$-$Q$ equation}
\label{app-inhom}

In this appendix, we transpose to the quasi-periodic 8-vertex model the characterization of the antiperiodic dynamical 6-vertex transfer matrix spectrum and eigenstates that has been obtained in \cite{LevNT15} through the solutions of an inhomogeneous $T$-$Q$ equation. This leads to a complete description of the spectrum in terms of Bethe-type equations with an extra inhomogeneous term, and of the eigenstates in terms of the multiple action of the operator $\bar{D}(\lambda)$ \eqref{bar-D} evaluated at the Bethe roots on some pseudo-vacuum state, a representation which deeply resembles --- except for the inhomogeneous term of the Bethe equations --- what is usually obtained through ABA.

As explained in \cite{LevNT15}, the idea is to obtained a continuous version of the discrete SOV characterization of the spectrum and eigenstates in the form of a $T$-$Q$ equation which admits $Q$-solutions of the same functional form as the usual functions $\mathsc{a}(\lambda)$, $\mathsc{d}(\lambda)$ and the transfer matrix eigenvalues $\mathsf{t}(\lambda)$ of the model, i.e. such that
\begin{equation}\label{Q-form}
  Q(\lambda )=\prod_{j=1}^{\mathsf{M}}\theta (\lambda -\lambda _{j}),
  \qquad \mathsf{M}\in\mathbb{N},\quad \lambda_1,\ldots,\lambda_\mathsf{N}\in\mathbb{C}.
\end{equation}
This can be done by considering the solutions of an adequately modified version of \eqref{hom-eq},
\begin{multline}\label{inhom}
\mathsf{t}(\lambda )\,Q(\lambda )
= (-1)^\mathsf{y}\,(-i)^\mathsf{xy}\, f(\lambda )\,\mathsc{a}(\lambda)\,Q(\lambda -\eta )
+(-1)^\mathsf{x}\, i^\mathsf{xy}\,\frac{\mathsc{d}(\lambda )}{f(\lambda +\eta )}\,  Q(\lambda +\eta )\\
-\mathsc{a}(\lambda)\, \mathsc{d}(\lambda )\, F(\lambda),
\end{multline}
with in particular a possible inhomogeneous term vanishing at the inhomogeneity and shifted inhomogeneity parameters $\xi_j-\eta h_j$, $j\in\{1,\ldots, \mathsf{N} \}$, $h_j\in\{0,1\}$.
Note that the introduction of the function $f(\lambda)$ is necessary so as to impose that all the terms in \eqref{inhom} with the choice \eqref{Q-form} have the same quasi-periodicity properties.
A somewhat minimal choice of $f(\lambda)$ when $\mathsf{M}=\mathsf{N}$ in \eqref{Q-form} is given by 
\begin{equation}
f(\lambda)\equiv f_{\mu }^{(\beta )}(\lambda )
=\beta ^{-1}e^{-i\mathsf{y}\lambda }\,\frac{\theta (\lambda -\mu )}{\theta (\lambda -\mu +t_0)},
\label{def-f}
\end{equation}
where $\beta $, $\mu $ are two arbitrary fixed complex parameters ($\beta\in\mathbb{C}\setminus\mathbb{R}$ and $\mu-\xi_j,\mu-\xi_j-\eta,\mu-t_0-\xi_j,\mu-t_0-\xi_j+\eta\notin\pi\mathbb{Z}+\pi\omega\mathbb{Z}$, $\forall\, j\in\{1,\ldots,\mathsf{N}\}$). With such a choice, the inhomogeneous term in \eqref{inhom} which is given by a function $F(\lambda)\equiv F_{\mu ,Q}^{(\beta )}(\lambda )$ which depends on $\beta$, $\mu$ and $Q(\lambda)$ as
\begin{multline}  \label{F-Q}
F_{\mu ,Q}^{(\beta )}(\lambda )
=\frac{\beta^{-1}\left( -1\right) ^{\mathsf{y}}(-i)^{\mathsf{xy}}\,e^{-i\mathsf{y}\lambda }\,\theta (t_0)}
          {\theta (t_0+\alpha _{Q}-\sum_{k}\xi _{k}+\mathsf{N}\eta )}
  \frac{Q(\mu -\eta -t_0)}{\mathsc{d}(\mu -t_0)}
  \frac{\theta (\lambda -\mu -\alpha _{Q}+\sum_{k}\xi _{k}-\mathsf{N}\eta )}{\theta(\lambda -\mu +t_0)} 
  \\
+\frac{\left( -1\right) ^{\mathsf{x}}i^{\mathsf{yx}}\beta \,e^{i\mathsf{y}(\lambda +\eta )}\,\theta (t_0)}
          {\theta (\mathsf{y}\pi \omega -t_0-\alpha _{Q}+\sum_{k}\xi _{k}-\mathsf{N}\eta )}
  \frac{Q(\mu )}{\mathsc{a}(\mu -\eta )}
  \frac{\theta (\lambda -\mu +\eta +\mathsf{y}\pi\omega -t_0-\alpha _{Q}+\sum_{k}\xi _{k}-\mathsf{N}\eta )}
         {\theta (\lambda -\mu +\eta )},
\end{multline}
with $\alpha _{Q}\equiv \sum_{j=1}^{\mathsf{N}}\lambda _{j}$ being the norm
of the theta function $Q(\lambda )$ of order $\mathsf{N}$.
Then, as it follows from the study of \cite{LevNT15} and from the correspondence established in the present paper between the antiperiodic SOS transfer matrix eigenstates and the 8-vertex $\mathsf{(x,y)}$-twisted ones, the description of the 8-vertex transfer matrix spectrum and eigenstates that we obtain through the consideration of this equation is complete. One of the advantages of this description is that it provides an ABA-type representation for the eigenstates, as stated below.

Indeed, let us introduce the following right and left pseudo-vacuum states,
\begin{align}
& \ket{\underline{\Omega} _{f} } 
=\sum_{\mathbf{h}\in \{0,1\}^{\mathsf{N}}}\prod_{a=1}^{\mathsf{N}}
\left( \frac{e^{i\mathsf{y}\eta }\,\mathsc{a}_{\mathsf{x,y}}(\xi_{a})}{\mathsc{d}(\xi _{a}-\eta )}\,f(\xi _{a})\right)^{\!h_{a}}
\det_{\mathsf{N}}\big[\Theta ^{(0,\mathbf{h})}\big]\,\mathbf{S}^{(0)}\ket{\underline{\mathbf{h}}} , 
 \label{refT-r-Q} \\
& \bra{ \underline{\Omega}_{f}}
=\sum_{\mathbf{h}\in \{0,1\}^{\mathsf{N}}}\prod_{a=1}^{\mathsf{N}}\Big(e^{i\mathsf{y}\eta }\,f(\xi _{a})\Big)^{h_{a}}
\det_{\mathsf{N}}\big[\Theta ^{(0,\mathbf{h})}\big]\,\bra{\underline{\mathbf{h}}}\big[\mathbf{S}^{(0)}\big]^{-1}.
\label{refT-l-Q}
\end{align}

\begin{theorem}\label{th-B1}
Let us suppose that the inhomogeneity parameters $\xi _{1},\ldots ,\xi _{\mathsf{N}}$ satisfy \eqref{cond-inh0}-\eqref{cond-inh} and that $\eta \in \mathbb{C}\setminus \mathbb{R}$.
Then, if $\mathsf{(x,y)}\not=(0,0)$, the spectrum $\Sigma _{\mathsf{(x,y)}}^{\mathsf{(8V)}}$ of the $\mathsf{(x,y)}$-twisted transfer matrix $\mathsf{T}_{\mathsf{(x,y)}}^{\mathsf{(8V)}}(\lambda )$ is given by the set of entire functions $\mathsf{t}(\lambda)$ for which there exists a function $Q(\lambda )$ of the form \eqref{Q-form} with $\mathsf{M}=\mathsf{N}$ satisfying $(Q(\xi _{j}),Q(\xi _{j}-\eta ))\not=(0,0)$, $1\leq j\leq \mathsf{N}$, and such that $\mathsf{t}(\lambda )$ and $Q(\lambda )$ satisfy the
inhomogeneous functional equation \eqref{inhom} with \eqref{def-f} and \eqref{F-Q}.
Moreover, the one-dimensional right and left eigenspaces in $\mathbb{V}_{%
\mathsf{N}}^{\mathcal{R/L}}$ corresponding to $\mathsf{t}(\lambda )\in
\Sigma _{\mathsf{(x,y)}}^{\mathsf{(8V)}}$ are respectively spanned by the vectors
\begin{equation}
   \ket{\Psi_\Lambda} =\prod_{a=1}^{\mathsf{N}}\bar{D}(\lambda _{a}) \ket{\underline{\Omega}_{f} } ,
   \qquad 
   \bra{ \Psi_\Lambda}=\bra{ \underline{\Omega}_{f}}\prod_{a=1}^{\mathsf{N}}\bar{D}(\lambda _{a}),
\end{equation}
where $\bar{D}(\lambda _{a})$ has been defined in \eqref{bar-D} and where $\Lambda\equiv\{\lambda _{1},\ldots,\lambda _{\mathsf{N}}\}$ is the set of  zeros of $Q(\lambda )$.
\end{theorem}

Theorem~\ref{th-B1} applies only to twisted 8-vertex transfer matrices, i.e. the case of $\mathsf{N}$ odd and periodic chain is not described above. However, this can be done just using Theorem~\ref{th-sp-8V-per} or Theorem~\ref{th-T8V-per} and the result of Appendix~B of \cite{LevNT15} (which applies also to the $\mathsf{x}=\mathsf{y}=0$ case). 
To this aim, we define two different pseudo-vacuum states $\ket{\underline{\Omega} _{f}^\epsilon }$ (or $\bra{ \underline{\Omega}_{f}^\epsilon}$) for $\epsilon=+,-$ as
\begin{align}
& \ket{\underline{\Omega} _{f}^\epsilon } 
=\sum_{\mathbf{h}\in \{0,1\}^{\mathsf{N}}}\prod_{a=1}^{\mathsf{N}}
\left( \frac{\epsilon\,\mathsc{a}(\xi_{a})}{\mathsc{d}(\xi _{a}-\eta )}\,f(\xi _{a})\right)^{\!h_{a}}
\det_{\mathsf{N}}\big[\Theta ^{(0,\mathbf{h})}\big]\,\mathbf{S}^{(+)}\ket{\underline{\mathbf{h}}} , 
 \label{ref-r-eps} \displaybreak[0]\\
& \bra{ \underline{\Omega}_{f}^\epsilon}
=\sum_{\mathbf{h}\in \{0,1\}^{\mathsf{N}}}\prod_{a=1}^{\mathsf{N}}\Big(\epsilon\,f(\xi _{a})\Big)^{h_{a}}
\det_{\mathsf{N}}\big[\Theta ^{(0,\mathbf{h})}\big]\,\bra{\underline{\mathbf{h}}}\big[\mathbf{S}^{(+)}\big]^{-1}.
\label{ref-l-eps}
\end{align}

\begin{theorem}
\label{th-B2}
Let us suppose that the inhomogeneity parameters $\xi _{1},\ldots ,\xi _{\mathsf{N}}$ satisfy \eqref{cond-inh0}-\eqref{cond-inh} and that $\eta \in \mathbb{C}\setminus \mathbb{R}$.
Then, the spectrum $\Sigma_{\mathsf{(0,0)}}^\mathsf{(8V)}$ of the periodic 8-vertex transfer matrix $\mathsf{T}^{\mathsf{(8V)}}_\mathsf{(0,0)}(\lambda )$ \eqref{transfer} for $\mathsf{N}$ odd is given by the set of entire functions $\mathsf{t}(\lambda)$ satisfying the condition \eqref{cond-+} and for which there exists a function $Q(\lambda )$ of the form \eqref{Q-form} with $\mathsf{M}=\mathsf{N}$ such that $(Q(\xi _{j}),Q(\xi _{j}-\eta ))\not=(0,0)$, $1\leq j\leq \mathsf{N}$, and such that $\mathsf{t}(\lambda )$ and $Q(\lambda )$ satisfy the
inhomogeneous functional equation \eqref{inhom} with \eqref{def-f} and \eqref{F-Q}.
A basis of the 
eigenspace associated with the eigenvalue $\mathsf{t}(\lambda ) \in \Sigma_{\mathsf{(0,0)}}^\mathsf{(8V)}$ is then provided by the two vectors  
%
\begin{equation}
   \ket{\Psi_\Lambda^\epsilon} =\prod_{a=1}^{\mathsf{N}}\bar{D}(\lambda _{a}) \ket{\underline{\Omega}_{f}^\epsilon } ,
   \qquad 
   \bra{ \Psi_\Lambda^\epsilon}=\bra{ \underline{\Omega}_{f}^\epsilon}\prod_{a=1}^{\mathsf{N}}\bar{D}(\lambda _{a}),
   \qquad \epsilon=+,-,
\end{equation}
where $\bar{D}(\lambda _{a})$ has been defined in \eqref{bar-D} and where $\Lambda\equiv\{\lambda _{1},\ldots,\lambda _{\mathsf{N}}\}$ is the set of  zeros of $Q(\lambda )$.
\end{theorem}



\providecommand{\bysame}{\leavevmode\hbox to3em{\hrulefill}\thinspace}
\providecommand{\MR}{\relax\ifhmode\unskip\space\fi MR }
\providecommand{\MRhref}[2]{%
  \href{http://www.ams.org/mathscinet-getitem?mr=#1}{#2}
}
\providecommand{\href}[2]{#2}

\end{document}